\documentclass[a4paper,reqno,twoside]{amsart}
\usepackage{bbm}
\usepackage[left=3.5cm,right=3.5cm,top=3.5cm,bottom=3.5cm,footskip=1.5cm,headsep=1cm,bindingoffset=0.cm,]{geometry}

\usepackage[utf8]{inputenc}
\usepackage[english]{babel}
\usepackage[T1]{fontenc} 
\usepackage{lmodern} 
\rmfamily 
\DeclareFontShape{T1}{lmr}{b}{sc}{<->ssub*cmr/bx/sc}{}
\DeclareFontShape{T1}{lmr}{bx}{sc}{<->ssub*cmr/bx/sc}{}

\numberwithin{equation}{section}
\usepackage{amsmath}
\usepackage{amssymb}
\usepackage{amsthm}
\usepackage{amsfonts}
\usepackage{mathtools}

\usepackage{mathdots}
\usepackage{caption}
\usepackage{subcaption}

\usepackage{dsfont} 
\usepackage[format=hang]{caption}
\usepackage[normalem]{ulem}

\usepackage{standalone}
\usepackage{graphicx}
\usepackage{imakeidx}
\makeindex
\usepackage{setspace}
\usepackage{appendix}
\usepackage{pdfpages}
\usepackage{upgreek}
\usepackage{tabulary}
\usepackage{booktabs}
\usepackage{extarrows}
\usepackage{tabularx}
\usepackage{float}
\restylefloat{table}
\usepackage{enumitem}
\usepackage{csquotes}
\usepackage{bm}
\usepackage{orcidlink}
\usepackage{lipsum}

\usepackage{tikz}
\usepackage{tikz-layers}
\usepackage{tikz-cd}
\usepackage[arrow, matrix, curve]{xy}
\usetikzlibrary{matrix,positioning,decorations.pathreplacing,calc}
\usetikzlibrary{
    decorations.text,%
    decorations.markings,%
    shadows}
\usetikzlibrary{calc,intersections}
\usepackage{adjustbox}
\usepackage{graphicx}

\usepackage[
    style=numeric,
    sorting=nty,
    maxnames=99,
    maxalphanames=5,
    natbib=true,
    backend=biber,
    sortcites]{biblatex}

\DeclareNameAlias{default}{family-given}

\AtEveryBibitem{
    \clearfield{url}
    \clearfield{issn}
    \clearfield{isbn}
    \clearfield{urldate}

    \ifentrytype{book}{
        \clearfield{pages}}{%
    }
}

\usepackage{xargs}                      
\usepackage[prependcaption,textsize=tiny,textwidth=3cm]{todonotes}
\newcommandx{\unsure}[2][1=]{\todo[linecolor=red,backgroundcolor=red!25,bordercolor=red,#1]{#2}}
\newcommandx{\change}[2][1=]{\todo[linecolor=blue,backgroundcolor=blue!25,bordercolor=blue,#1]{#2}}
\newcommandx{\info}[2][1=]{\todo[linecolor=OliveGreen,backgroundcolor=OliveGreen!25,bordercolor=OliveGreen,#1]{#2}}
\newcommandx{\improvement}[2][1=]{\todo[linecolor=black,backgroundcolor=black!25,bordercolor=black,#1]{#2}}
\newcommandx{\thiswillnotshow}[2][1=]{\todo[disable,#1]{#2}}

\usepackage[noabbrev,capitalize]{cleveref}
\crefname{proposition}{Proposition}{Propositions}
\crefname{equation}{}{}

\newtheorem{theorem}{Theorem}[section]
\newtheorem{lemma}[theorem]{Lemma}
\newtheorem{proposition}[theorem]{Proposition}
\newtheorem{corollary}[theorem]{Corollary}

\theoremstyle{definition}
\newtheorem{definition}[theorem]{Definition}

\newtheorem{remark}[theorem]{Remark}

\crefname{assumption}{Assumption}{Assumptions}
\crefname{definition}{Definition}{Definitions}
\crefname{corollary}{Corollary}{Corollaries}
\crefname{enumi}{item}{items}

\creflabelformat{subfigure}{#2\textsc{#1}#3}

\usepackage[final]{microtype}

\makeatletter
\newsavebox\myboxA
\newsavebox\myboxB
\newlength\mylenA

\newcommand*\xoverline[2][0.75]{%
  \sbox{\myboxA}{$\m@th#2$}%
  \setbox\myboxB\null
  \ht\myboxB=\ht\myboxA%
  \dp\myboxB=\dp\myboxA%
  \wd\myboxB=#1\wd\myboxA
  \sbox\myboxB{$\m@th\overline{\copy\myboxB}$}
  \setlength\mylenA{\the\wd\myboxA}
  \addtolength\mylenA{-\the\wd\myboxB}%
  \ifdim\wd\myboxB<\wd\myboxA%
    \rlap{\hskip 0.5\mylenA\usebox\myboxB}{\usebox\myboxA}%
  \else
    \hskip -0.5\mylenA\rlap{\usebox\myboxA}{\hskip 0.5\mylenA\usebox\myboxB}%
  \fi}
\makeatother

\usepackage{fancyhdr}

\fancypagestyle{plain}{%
    \fancyhead[C]{}
    \fancyfoot[C]{}
    
    }
\fancypagestyle{nosection}{%
    \fancyhead[CE]{}
    \fancyhead[CO]{}
    \fancyfoot[CE,CO]{\thepage}
}

\DeclareMathOperator{\dom}{dom}

\DeclareMathOperator{\SL}{SL}

\usepackage{tikz}

\usetikzlibrary{matrix} 
\usetikzlibrary{arrows,arrows.meta,bending} 

\usepackage{nicematrix}

\usetikzlibrary{hobby}

\usepackage{calc}

\newcommand{\seq}[2]{{#1^{(#2)}}}

\newcommand{\hs}{L^2(\R, \dd\mu)}
\newcommand{\tp}{\Tilde{P}}

\DeclareMathOperator{\Z}{\mathbb{Z}}

\DeclareMathOperator{\R}{\mathbb{R}}
\DeclareMathOperator{\C}{\mathbb{C}}

\renewcommand{\i}{\mathbf{i}}

\renewcommand{\qedsymbol}{$\blacksquare$}

\newcommand{\nm}{\noalign{\smallskip}}
\newcommand{\ds}{\displaystyle}

\renewcommand{\epsilon}{\varepsilon}
\DeclareMathOperator{\dd}{\mathrm{d}}

\renewcommand{\i}{\mathbf{i}}

\DeclareMathOperator{\iL}{{\mathsf{L}}}
\DeclareMathOperator{\iR}{{\mathsf{R}}}
\DeclareMathOperator{\iLR}{{\mathsf{L},\mathsf{R}}}

\usepackage{tcolorbox}
\usetikzlibrary{tikzmark,calc}

\newcommand{\ip}[2]{\left\langle #1, #2 \right\rangle}
\newcommand{\mc}[1]{\mathcal{#1}}

\newcommand{\abs}[1]{\left\lvert#1\right\rvert}

\newcommand{\norm}[1]{\left\lVert#1\right\rVert}

\newcommandx{\silvio}[2][1=]{\todo[linecolor=blue,backgroundcolor=blue!25,bordercolor=blue,#1]{Silvio: #2}}

\newcommandx{\bowen}[2][1=]{\todo[linecolor=blue,backgroundcolor=blue!25,bordercolor=blue,#1]{bowen: #2}}

\newcommandx{\alex}[2][1=]{\todo[linecolor=red,backgroundcolor=red!25,bordercolor=red,#1]{Alex: #2}}

\newcommandx{\jiayu}[2][1=]{\todo[linecolor=red,backgroundcolor=red!25,bordercolor=red,#1]{Jiayu: #2}}

\title[Non-Hermitian Fabry--Pérot Resonances]{Non-Hermitian Fabry--Pérot Resonances}

\begin{document}
 \author[H. Ammari]{Habib Ammari \,\orcidlink{0000-0001-7278-4877}}
 \address{\parbox{\linewidth}{Habib Ammari\\
  ETH Z\"urich, Department of Mathematics, Rämistrasse 101, 8092 Z\"urich, Switzerland, \href{http://orcid.org/0000-0001-7278-4877}{orcid.org/0000-0001-7278-4877}}.}
  \email{habib.ammari@math.ethz.ch}
  \thanks{}

 \author[E.O. Hiltunen]{Erik Orvehed Hiltunen\,\orcidlink{0000-0003-2891-9396}}
\address{\parbox{\linewidth}{Erik Orvehed Hiltunen\\
Department of Mathematics, University of Oslo, Moltke Moes vei 35, 0851 Oslo, Norway, 
 \href{http://orcid.org/0000-0003-2891-9396}{orcid.org/0000-0003-2891-9396}}}
\email{erikhilt@math.uio.no}

\author[B. Li]{Bowen Li} 
\address{\parbox{\linewidth}{Bowen Li\\ Department of Mathematics,  City University of Hong Kong, Kowloon Tong, Hong Kong SAR}}
\email{bowen.li@cityu.edu.hk}

\author[P. Liu]{Ping Liu}
\address{\parbox{\linewidth}{Ping Liu\\
School of Mathematical Sciences and Institute of
Fundamental and Transdisciplinary Research, Zhejiang University, Hangzhou, 310027, China.}}
\email{pingliu@zju.edu.cn}

\author[J. Qiu]{Jiayu Qiu} 
\address{\parbox{\linewidth}{Jiayu Qiu\\
Department of Mathematics, ETH Z\"{u}rich, R\"{a}mistrasse 101, CH-8092 Z\"{u}rich, Switzerland}}
\email{jiayu.qiu@sam.math.ethz.ch}

 \author[Y. Shao]{Yingjie Shao}
 \address{\parbox{\linewidth}{Yingjie Shao\\
School of Mathematical Sciences, Zhejiang University, Hangzhou, 310027, China.}}
\email{shaoyingjie323@zju.edu.cn}

 \author[A. Uhlmann]{Alexander Uhlmann\,\orcidlink{0009-0002-0426-6407}}
   \address{\parbox{\linewidth}{Alexander Uhlmann\\
   ETH Z\"urich, Department of Mathematics, Rämistrasse 101, 8092 Z\"urich, Switzerland, \href{http://orcid.org/0009-0002-0426-6407}{orcid.org/0009-0002-0426-6407}}.}
  \email{alexander.uhlmann@sam.math.ethz.ch}

 \begin{abstract}
 We characterise non-Hermitian Fabry--Pérot resonances in high-contrast resonator systems and study the properties of their associated resonant modes from continuous differential models. We consider two non-Hermitian effects: the exceptional point degeneracy and the skin effect induced by imaginary gauge potentials. Using a propagation matrix formalism, we characterise these two non-Hermitian effects beyond the subwavelength regime. This analysis allows us to (i) establish the existence of exceptional points purely from radiation conditions and to (ii) prove that the non-Hermitian skin effect applies uniformly across resonant modes, yielding broadband edge localisation.
 \end{abstract}

\maketitle

 \par

\bigskip


\section{Introduction and motivation}
In this paper, we consider continuous models for systems of non-Hermitian high-contrast resonators and characterise their Fabry--P\'erot resonances. In the subwavelength regime, these scattering resonances (also called resonant frequencies) and effects obtained coincide with those derived earlier through the capacitance matrix formalism \cite{ammari.davies.ea2024Functional}. Two non-Hermitian effects are studied beyond the subwavelength regime: the approximate and exact exceptional points in systems of resonators and the non-Hermitian skin effect induced by imaginary gauge potentials. 

Exceptional points typically occur when a system's resonant frequencies and associated eigenmodes simultaneously coalesce as the system parameters are varied. In terms of finite, discrete systems described by a matrix, they correspond to parameter configurations where the system ceases to be diagonalisable. They are a unique feature of non-Hermitian systems and offer promising applications to enhance sensing due to the higher-order dependence on perturbations of systems at exceptional points \cite{exceptpoint2,hodaei.hassan.ea2017Enhanced, vollmer.arnold.ea2008Single}.

The non-Hermitian skin effect is the phenomenon whereby -- as a consequence of non-reciprocity -- a large proportion 
of the bulk eigenmodes of a non-Hermitian system are localised at one edge of the system. It allows for unidirectional wave control, offering wide-ranging applications in topological photonics, phononics and other condensed matter systems \cite{hatano.nelson1996Localization,yokomizo,rivero.feng.ea2022Imaginary}.

In the first part of the paper, we consider resonator systems with potentially complex wave speeds inside the resonators (in the non-subwavelength regime) and characterise their exceptional points in terms of a mismatch between the order of the zero $\omega_0$ of the characteristic determinant obtained from the total propagation matrix and the dimension of the space of solutions of the scattering resonance problem at $\omega_0$. By doing so, we recover the results obtained in the subwavelength regime that are based on the non-diagonalisability of the capacitance matrix associated with the resonator system \cite{exceptpoint1,exceptpoint2}, which requires the use of complex wave speeds. Both definitions coincide in the subwavelength regime. They also coincide with the fact that a resonant frequency for which exceptional points occur is a pole of the Green function associated to the system that is of order higher than one. This can be seen in the non-subwavelength regime from the fact that the Green function, as the fundamental solution of \eqref{eq: gen Sturm-Liouville}, can be expressed in terms of propagation matrices \cite{SANCHEZSOTO2012191}. In the subwavelength regime, the (discrete) Green function at a frequency $\omega$ is nothing more than the inverse of $\omega^2 I - \Tilde{\mathcal{C}}(0)$ with $\Tilde{\mathcal{C}}(0)$ being the capacitance matrix of the resonator system \cite{ammari.davies.ea2024Functional}. 
Most of the exceptional points discussed in this part are approximate in the sense that they are achieved by ensuring the coalescence of the leading-order behaviour of resonances as the contrast material parameter goes to zero. Nevertheless, by changing the radiation conditions into parity-time symmetric ones, we show analytically that such systems with balanced energy gain and loss have exact exceptional points and demonstrate numerically the square-root coalescence to parameters around both the approximate and exact exceptional points.

The second part of this paper is devoted to the study of the non-Hermitian skin effect in the non-subwavelength regime. Using again propagation matrices for the system and in this non-reciprocal case their symmetrisation, we characterise the scattering resonances of the Helmholtz equation with imaginary gauge potential inside the resonators and show that the skin effect holds even for non-subwavelength modes. The eigenmodes of the resonator system still exhibit a uniform decay that is  independent of frequency. Nevertheless, compared to the subwavelength regime, they are not constant inside the resonators; instead, they may highly oscillate. Note that in \cite{ammari.barandun.ea2023NonHermitian} the occurrence of the non-Hermitian skin effect in the presence of an imaginary gauge potential is proved using Toeplitz matrix theory for a discrete approximation of the continuous model (which is based on differential equations) in terms of the gauge capacitance matrix that holds only in the subwavelength regime. In this paper we prove the occurrence of the non-Hermitian skin effect for the continuous problem both in the subwavelength and non-subwavelength regimes. 

In fact, our results in this paper constitute the first results on non-Hermitian scattering resonances for systems of resonators beyond the subwavelength regime. They give an analytical and numerical framework for computing the scattering resonances in two non-Hermitian settings. In particular, they show that non-Hermitian phenomena  (exceptional degeneracies and accumulation of eigenmodes at one edge of the system) that are known to occur in the subwavelength regime hold also in the non-subwavelength regime. The use of the new discrete approximation in terms of the generalised capacitance matrix in the non-subwavelength regime that is introduced in \cite{pm2} allows for a precise 
description of the resonances in the non-Hermitian settings considered here, how many there are, and the asymptotic properties of their corresponding resonant states. It is expected that such a characterisation holds in three dimensions,  which would also lead to the study of non-Hermitian phenomena in three-dimensional systems beyond the subwavelength regime directly from the continuous model. 

Moreover, while this paper considered simple finite or infinite periodic geometries, another open question is the characterisation of bandgaps for spatially disordered resonator arrays beyond the subwavelength regime. We expect that adapting \cite{ammari.thalhammer.ea2025Uniform} to the propagation matrix framework used in this work should allow for the guaranteed existence of bandgaps from local properties.

Our paper is organised as follows. In \Cref{psetting}, we introduce the scattering resonance problem and its capacitance matrix formulation that hold beyond the subwavelength regime. Then, in \Cref{sec:ep}, we elucidate the variety of exceptional points (exact and asymptotic) that arise in the non-Hermitian setting. Namely, in \cref{ssec:complex_ep} we characterise exceptional points induced by leading-order Hermiticity in the form of complex material parameters. In \cref{ssec:rad_ep}, we demonstrate how exceptional points may be obtained from the non-Hermiticity due to outgoing radiation. Moreover, in \cref{ssec:prc_ep}, we demonstrate that a $\mc{PT}$-symmetric modified radiation condition allows for explicit exceptional points.
\Cref{sec:NR} is devoted to the study of the scattering resonances in non-reciprocal systems and to the proof of the non-Hermitian skin effect in such systems by means of a symmetrisation approach. By combining our approaches in this paper with those in \cite{junshan1, ammari.barandun.ea2025Generalized,barandun2023}, we expect to generalise our results in this paper to periodic infinite (for not only the exceptional point degeneracy but also the Dirac degeneracy) and semi-infinite (for the non-Hermitian skin effect) structures.

\section{Setting and capacitance matrix formalism} \label{psetting}

We consider a one-dimensional chain of $N$ disjoint resonators $D_i\coloneqq (x_i^{\iL},x_i^{\iR})$, where $(x_i^{\iLR})_{1\<i\<N} \subset \R$ are the $2N$ extremities satisfying $x_i^{\iL} < x_i^{\iR} <  x_{i+1}^{\iL}$ for any $1\leq i \leq N$. We denote by $\ell_i = x_i^{\iR} - x_i^{\iL}$ the length of each of the resonators and by $s_i= x_{i+1}^{\iL} -x_i^{\iR}$ the spacing between the $i$\textsuperscript{th} 
and $(i+1)$\textsuperscript{th} resonators. We denote by
\begin{align} \label{eq:config}
   \mc D \coloneqq \bigcup_{i=1}^N D_i =  \bigcup_{i=1}^N(x_i^{\iL},x_i^{\iR})
\end{align}
the set of resonators; see \cref{fig:settingchap2}.

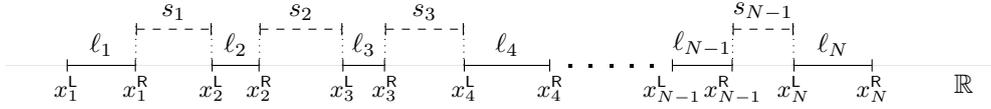
\begin{figure}[!h]
        \centering
        \begin{tikzpicture}[scale=0.8]
            \pgfmathsetseed{8}
            \coordinate (0) at (0,0);
            \foreach \i in {1,2,...,7}{
                \pgfmathsetmacro{\r}{max(2*rnd,0.7)}
                \pgfmathsetmacro{\im}{\i-1}
                \coordinate (\i) at ($(\im)+(\r,0)$);
           };
            \foreach \i in {0,2,...,7}{
                \pgfmathsetmacro{\im}{\i}
                \pgfmathsetmacro{\ip}{\i+1}
                \pgfmathsetmacro{\lin}{int(\i/2+1)}
                \draw[|-|] (\im) -- (\ip) ; 
                \draw ($0.5*(\im)+0.5*(\ip)$) node[above] {$\ell_{\lin}$};
                \draw node[below,yshift=-0.2] at (\im) {\small $x_{\lin}^{\iL}$}; 
                \draw node[below,yshift=-0.2] at (\ip) {\small $x_{\lin}^{\iR}$}; 
            };
            \foreach \i in {1,3,...,6}{
                \pgfmathsetmacro{\im}{\i}
                \pgfmathsetmacro{\ip}{\i+1}
                \pgfmathsetmacro{\lin}{int((\i-1)/2+1)}
                \pgfmathsetmacro{\linp}{int((\i-1)/2+2)}
                \draw[dotted] (\im) -- ($(\im)+(0,0.7)$);
                \draw[dotted] (\ip) -- ($(\ip)+(0,0.7)$);
                \draw[|-|,dashed] ($(\im)+(0,0.6)$) -- ($(\ip)+(0,0.6)$);
                \draw node[above] at ($0.5*(\im)+0.5*(\ip)+(0,0.6)$) {$s_{\lin}$};
            }       
            \coordinate (8) at ($(7)+(2,0)$);
            \coordinate (9) at ($(7)+(3,0)$);
            \coordinate (C) at ($0.5*(7)+0.5*(8)$);
            \draw[|-|] (8) -- (9) ; 
            \draw node[below,yshift=-0.2] at (8) {\small $x_{N-1}^{\iL}$}; 
            \draw node[below,yshift=-0.2] at (9) {\small $x_{N-1}^{\iR}$}; 
            \draw ($0.5*(8)+0.5*(9)$) node[above] {$\ell_{N-1}$};
            \coordinate (10) at ($(9)+(1,0)$);
            \coordinate (11) at ($(9)+(2.3,0)$);
            \draw[|-|] (10) -- (11) ; 
            \draw node[below,yshift=-0.2] at (10) {\small $x_{N}^{\iL}$}; 
            \draw node[below,yshift=-0.2] at (11) {\small $x_{N}^{\iR}$}; 
            \draw ($0.5*(10)+0.5*(11)$) node[above] {$\ell_{N}$};
            \draw[dotted] (9) -- ($(9)+(0,0.7)$);
            \draw[dotted] (10) -- ($(10)+(0,0.7)$);
            \draw[|-|,dashed] ($(9)+(0,0.6)$) -- ($(10)+(0,0.6)$);
            \draw node[above] at ($0.5*(9)+0.5*(10)+(0,0.6)$) {$s_{N-1}$};
            \foreach \i in {-2,-1,0,1,2} {
                \pgfmathsetmacro\c{0.03}
                \coordinate (a) at ($(C)+(\i*1/3,0)$);
                \fill ($(a)-(\c,\c)$)  rectangle ($(a)+(\c,\c)$);};
            \draw[color=black,draw opacity=0.1] ($(0)-(1,0)$) -- (7);
            \draw[color=black,draw opacity=0.1] (8) -- ($(11)+(2,0)$);
            \draw node[below] at ($(11)+(1.5,0)$) {$\R$};
        \end{tikzpicture}
        \caption{A one-dimensional chain of $N$ disjoint resonators.}
        \label{fig:settingchap2}
    \end{figure}

In the first part of this work, we consider the one-dimensional Helmholtz equation: 
\begin{align}
    -\frac{\omega^{2}}{\kappa(x)}u(x)-\frac{\dd}{\dd x}\left( \frac{1}{\rho(x)}\frac{\dd u}{\dd
    x}  (x)\right) =0,\qquad x \in\R,
    \label{eq:HH}
\end{align}
subject to the outgoing radiation conditions:
\begin{align}
\label{eq:SRC}
\frac{\dd u}{\dd \abs{x}}(x) -\i \frac{\omega}{v} u(x) = 0, & \qquad x\in(-\infty,x_1^{\iL})\cup (x_N^{\iR},\infty).
\end{align}
The material parameters are piecewise constant and given by
\begin{align*}
    \kappa(x)=
    \begin{dcases}
        \kappa_i & x\in D_i,\\
        \kappa&  x\in\R\setminus \mc D,
    \end{dcases}\quad\text{and}\quad
    \rho(x)=
    \begin{dcases}
        \rho_b & x\in D_i,\\
        \rho&  x\in\R\setminus \mc D,
    \end{dcases}
\end{align*}
for some constants $\rho_b, \rho, \kappa_i, \kappa$. 
Moreover, we denote by $v_i$ and $v$ the wave speeds inside the resonators $D_i$ and inside the background medium $\R\setminus \mc D$, respectively,  and let 
$r_i$ be the contrast between the wave speeds $v$ and $v_i$. We also denote by  $\delta$ the contrast between the densities of the resonators and the background medium. Therefore,
\begin{align} \label{eq:sympara}
    v_i:=\sqrt{\frac{\kappa_i}{\rho_b}}, \qquad v:=\sqrt{\frac{\kappa}{\rho}}, \qquad
    \delta:=\frac{\rho_b}{\rho}, \qquad r_i= \frac{v}{v_i}.
\end{align}
Throughout this paper, we assume that $\delta$ and $v$ are positive, while the wave speeds $v_i$ in the resonators may be \emph{complex}. 

Since the material parameters are piecewise constants, equation \eqref{eq:HH} together with \eqref{eq:SRC} can be written as 
\begin{align}\label{eq:HH_coupled}
    \begin{dcases}
    \frac{{\dd}^{2} u}{\dd x^2}(x) +\frac{\omega^2}{v_i^2}u(x)=0, & x\in D_i,\\
    \frac{\dd^2 u}{\dd x^2}(x)  + \frac{\omega^2}{v^2}u(x)=0, & x\in\R\setminus \mc D,\\
        u\vert_{\iR}(x^{\iLR}_i) - u\vert_{\iL}(x^{\iLR}_i) = 0, & \text{for all } 1\leq i\leq N,\\
        \left.\frac{\dd u}{\dd x}\right\vert_{\iR}(x^{\iL}_{{i}})=\delta\left.\frac{\dd u}{\dd x}\right\vert_{\iL}(x^{\iL}_{{i}}), & \text{for all } 1\leq i\leq N,\\
        \delta\left.\frac{\dd u}{\dd x}\right\vert_{\iR}(x^{\iR}_{{i}})=\left.\frac{\dd u}{\dd x}\right\vert_{\iL}(x^{\iR}_{{i}}), & \text{for all } 1\leq i\leq N,
    \end{dcases}
\end{align}
subject to the following boundary conditions:
\begin{equation} \label{eq:RadiationBC}
\left. \frac{\dd u}{\dd x} \right \vert_{\iL}(x_1^{\iL}) +\i \frac{\omega}{v} u(x_1^{\iL}) = 0, \quad \left. \frac{\dd u}{\dd x} \right \vert_{\iR}(x_N^{\iR}) -\i \frac{\omega}{v} u(x_N^{\iR}) = 0.
\end{equation}
Here, for a one-dimensional function $w$, we denote by
\begin{align*}
    w\vert_{\iL}(x) \coloneqq \lim_{\substack{s\to 0\\ s>0}}w(x-s) \quad \mbox{and} \quad  w\vert_{\iR}(x) \coloneqq \lim_{\substack{s\to 0\\ s>0}}w(x+s), 
\end{align*}
if the limits exist. 

We investigate the scattering resonances (or the resonant frequencies) in the \emph{high contrast regime} where
\[
    \delta \to 0 \quad \text{while} \quad v, v_i, s_i, \ell_i =\mc O(1).
\] 
The large impedance mismatch due to the high contrast ratio creates very sharp transmission peaks, ensuring that our scattering resonances are of \emph{Fabry-Pérot-type}.
To understand this, we introduce the propagation matrix formalism. Physically, for the given $i$\textsuperscript{th} resonator, the interior propagation matrix $P^\omega_{i,int}$ propagates the solution $u$ from the left interior edge to the right. 
Throughout, we interchangeably use the superscript $'$ and $d/dx$ to denote the derivative with respect to $x$. 

In terms of the Dirichlet and Neumann data $(u(x), u'(x))^\top$, $P^\omega_{i,int}$ satisfies
\begin{equation} \label{def:propga}
\Big(u(x_{i}^{\iR})\big|_{\iL}, u'(x_{i}^{\iR})\big|_{\iL}\Big)^\top = P^\omega_{i,int} \Big(u(x_i^{\iL})\big|_{\iR}, u'(x_{i}^{\iL})\big|_{\iR}\Big)^\top,
\end{equation}
where the superscript $\top$ denotes the transpose and 
\begin{equation*}
P^\omega_{i,int}=
\begin{pmatrix}
\cos(\frac{\omega}{v_i}\ell_i) & \frac{v_i}{\omega}\sin(\frac{\omega}{v_i}\ell_i) \\
-\frac{\omega}{v_i}\sin(\frac{\omega}{v_i}\ell_i) & \cos(\frac{\omega}{v_i}\ell_i)
\end{pmatrix} .
\end{equation*}
Similarly, the exterior propagation matrix propagates the wave across the $i$\textsuperscript{th} gap between the resonators as follows:
\begin{equation*}
\Big(u(x_{i+1}^{\iL})\big|_{\iL}, u'(x_{i+1}^{\iL})\big|_{\iL}\Big)^\top = P^\omega_{i,ext} \Big(u(x_i^{\iR})\big|_{\iR}, u'(x_{i}^{\iR})\big|_{\iR}\Big)^\top,
\end{equation*}
with
\begin{equation*}
P^\omega_{i,ext}=
\begin{pmatrix}
\cos(\frac{\omega}{v}s_i) & \frac{v}{\omega}\sin(\frac{\omega}{v}s_i) \\
-\frac{\omega}{v}\sin(\frac{\omega}{v}s_i) & \cos(\frac{\omega}{v}s_i)
\end{pmatrix} .
\end{equation*}
To account for the boundary conditions between the resonator and the gap through the contrast parameter $\delta$ in \eqref{eq:HH_coupled}, we introduce the following matrix:
\begin{equation*}
P_{ext\to int}:=\begin{pmatrix}
1 & 0 \\ 0 & \delta
\end{pmatrix}.
\end{equation*}
Putting these ingredients together, the propagation matrix $P^\omega_i \in \SL(2,\mathbb{C})$ for the $i$\textsuperscript{th} block (\emph{i.e.}, the  $i$\textsuperscript{th} resonator plus the $i$\textsuperscript{th} gap) is given by
\begin{equation} \label{pomegai}
P^\omega_{i}:=P^\omega_{i,ext}P_{ext\to int}^{-1}P^\omega_{i,int}P_{ext\to int},
\end{equation}
and satisfies the following identity:
\begin{equation*}
\Big(u(x_{i+1}^{\iL})\big|_{\iL}, u'(x_{i+1}^{\iL})\big|_{\iL}\Big)^\top = P^\omega_{i} \Big(u(x_i^{\iL})\big|_{\iL}, u'(x_{i}^{\iL})\big|_{\iL}\Big)^\top .
\end{equation*}
Here, $\SL(2,\mathbb{C})$ denotes the set of $2$ by $2$ complex matrices with determinant one.  The propagation matrix $P^\omega_{i}$ has the following expression:
{
\small\begin{align*}
   \begin{psmallmatrix}
      \ds  \cos(\frac{\omega}{v} s_i)\cos(\frac{\omega}{v_i} \ell_i)- \frac{r_i}{\delta} \sin(\frac{\omega}{v} s_i)  \sin(\frac{\omega}{v_i} \ell_i) & \ds
          \frac{v}{\omega} \cos(\frac{\omega}{v_i} \ell_i)\sin(\frac{\omega}{v} s_i) +  \frac{v_i}{\delta \omega}  \cos(\frac{\omega}{v} s_i) \sin(\frac{\omega}{v_i} \ell_i) \\
            \nm
         \ds   -\frac{\omega}{v}\cos(\frac{\omega}{v_i} \ell_i)\sin(\frac{\omega}{v} s_i) -\frac{\omega}{\delta v_i}\,\sin (\frac{\omega}{v_i} \ell_i) \cos (\frac{\omega}{v} s_i) & \ds
            \cos(\frac{\omega}{v} s_i)\cos(\frac{\omega}{v_i} \ell_i)
            -\frac{1}{\delta r_i} \sin (\frac{\omega}{v} s_i) \sin(\frac{\omega}{v_i} \ell_i)
        \end{psmallmatrix}.
\end{align*}}%
In the case $i=N$, the separation distance $s_N$ \emph{after} the final resonator is not well-defined. In fact, for the later characterisation of resonant frequencies, this choice is completely arbitrary. To still define $P_N^\omega$ in a consistent manner, we choose $s_N = 0$, which causes $P^\omega_{N, ext}$ to be equal to the identity matrix.  

One of the key advantages of the propagation matrix approach is that $P^\omega_{i}$ depends only on the material parameters of the $i$\textsuperscript{th} block, which makes many analytical approaches significantly more tractable and natural. 
Based on \eqref{pomegai}, we define the total propagation matrix $P^\omega_{tot}$ for the configuration \eqref{eq:config} as
\begin{equation} \label{ptotal}
P^\omega_{tot} \coloneqq P_{N}^\omega\cdots P_1^\omega.
\end{equation}

By matching the boundary conditions \eqref{eq:RadiationBC}, we arrive at the following characterisation of the resonant frequencies:
\begin{equation} \label{eq:characterisationExp}
    \omega \text{ is a resonant frequency of }\cref{eq:HH_coupled} \iff 
\det\left(P_{tot}^{\omega}
\begin{pmatrix}
    1 \\ -\i \frac{\omega}{v}
\end{pmatrix}\middle|
\begin{pmatrix}
    1 \\ \i \frac{\omega}{v}
\end{pmatrix}
\right) =0.
\end{equation}

\begin{definition}[Characteristic determinant]
For a finite resonator array, a key object is the \emph{characteristic determinant}
    \begin{equation} \label{eq:chdet}
        f(\omega;\delta):= \mathrm{det}\left(P_{tot}^{\omega}
    \begin{pmatrix}
        1 \\ -\i \frac{\omega}{v}
    \end{pmatrix}\middle|
    \begin{pmatrix}
        1 \\ \i \frac{\omega}{v}
    \end{pmatrix}
    \right).
    \end{equation}
\end{definition}

As $\delta\to 0$, the characteristic determinant converges uniformly on any compact set in $\omega$. Moreover, the following characterisation of the zeros of $f(\omega;0)$ holds; 
see \cite[Lemma 3.5]{pm1}.
\begin{lemma}\label{lem:deltazero_resonances}
    The zeros of $f(\omega;0)$ are given by the countable and discrete set
\[
E = \left\{\omega \mid \ell_j\frac{\omega}{v_j} \in \pi\Z\,,\ 1 \le j \le N \right\} \bigcup \left\{\omega \mid s_j\frac{\omega}{v} \in \pi\Z\,,\ 1 \le j \le N - 1 \right\}.
\]
Moreover, the order of a zero $\omega \in E$ is given by
\begin{align*}
    n(\omega) = \#\left\{j \mid \ell_j\frac{\omega}{v_j} \in \pi\Z\,,\ 1\leq j\leq N\right\} +  \#\left\{j \mid  s_j\frac{\omega}{v} \in \pi\Z\,,\ 1 \le j \le N - 1 \right\}.
\end{align*}
\end{lemma}
This result justifies the transmission comb observed in \cref{fig:dimer_transmission_comb} for a dimer system ($N=2$) with $\ell_1=\ell_2 = s=1$ and $v = v_i = 1$. The threefold zeros of $f(\omega; 0)$ around $\omega=0,\pi$ split into three distinct zeros with small imaginary parts for small $\delta>0$. We let  the transmission coefficient $T(\omega)$ for real $\omega$ be defined by
$$
T(\omega):= |u(x_N^{\iR})|,
$$
where $u$ is the solution to the scattering problem \eqref{eq:HH_coupled} with the boundary conditions
$$
\frac{\dd u}{\dd x}(x_1^{\iL}) +\i \frac{\omega}{v} u(x_1^{\iL}) = 2 \i \frac{\omega}{v} e^{\i \frac{\omega}{v} x}, \quad \frac{\dd u}{\dd x}(x_N^{\iR}) -\i \frac{\omega}{v} u(x_N^{\iR}) = 0.
$$
\cref{fig:dimer_transmission_comb} exhibits the combs with $2N-1$ peaks for $T(\omega)$ that are typical for Fabry-Pérot-type resonances.

\begin{figure}[!h]
    \centering
    \includegraphics[width=0.6\linewidth]{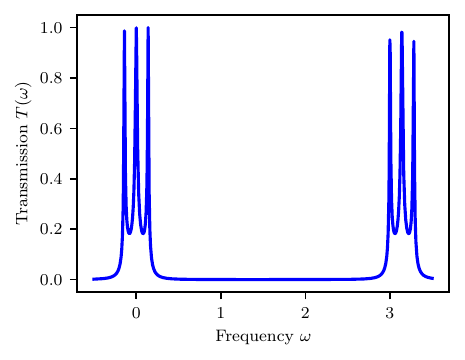}
    \caption{The transmission coefficient $T(\omega)$ for a dimer system ($N=2$, $\ell_i = v_i = v = s = 1$, $\delta = 10^{-1}$) exhibits transmission peaks forming a characteristic Fabry-Pérot-type comb-like structure.}
    \label{fig:dimer_transmission_comb}
\end{figure}

In the asymptotic regime, when $\delta \rightarrow 0$, the leading order behaviour of these scattering resonances is described by discrete capacitance matrix formulations \cite{pm1, pm2}. We let the vectors $\mathfrak{\bm t}=(\mathfrak{t}_1,\cdots,\mathfrak{t}_{2N-1})^\top$ and $\bm t = (t_1, \dots, t_{2N-1})^\top$ be defined by 
$$
    \mathfrak{\bm t}:=(r_1\ell_1,\ s_{1},\ r_2\ell_{2},\ s_{2},\ \cdots,\ r_{N-1}\ell_{N-1},\ s_{N-1},\ r_N\ell_N)^\top\in\R^{2N-1}_{>0},
    $$
    and
$$
    \bm t:=(r_1^2\ell_1,\ s_{1},\ r_2^2\ell_{2},\ s_{2},\ \cdots,\ r_{N-1}^2\ell_{N-1},\ s_{N-1},\ r_N^2\ell_N)^\top\in\R^{2N-1}_{>0}.
    $$
 For $\omega \in \R$ and $0 \leq j \leq 2N$, we introduce the coefficients
\[
c_j(\omega) = 
\begin{cases}
\dfrac{1}{t_j} & \text{if } \pi \mid \mathfrak{t}_j \frac{\omega}{v}, \\
0 & \text{otherwise},
\end{cases}
\quad \text{for } 1\leq j \leq 2N-1,
\]
and $c_0(\omega) = c_{2N}(\omega) = 0$. Here, $a\mid b$ means that $b/a$ is an integer. 
We let $$\theta_j(\omega) := c_j(\omega) c_{j+1}(\omega)  \qquad \text{for } 1\leq j \leq 2N-2,$$ and define the generalised capacitance matrix as the following $N \times N$ tridiagonal matrix:
\begin{align}\label{eq:mathcalCk_def}
\Tilde{\mc C}(\omega) := 
\begin{psmallmatrix}
    \theta_1(\omega) & -\theta_1(\omega) & & & \\
    -\theta_2(\omega) & \theta_2(\omega) + \theta_3(\omega) & -\theta_3(\omega) & & \\
    & -\theta_4(\omega) & \theta_4(\omega) + \theta_5(\omega) & -\theta_5(\omega) & \\
    & & \ddots & \ddots & \ddots \\
    & & & -\theta_{2N-4}(\omega) & \theta_{2N-4}(\omega) + \theta_{2N-3}(\omega) & -\theta_{2N-3}(\omega) \\
    & & & & -\theta_{2N-2}(\omega) & \theta_{2N-2}(\omega)
\end{psmallmatrix}.
\end{align}

The following result is from \cite{pm2}.
\begin{lemma}\label{thm:leadingorder}
    Let $\omega_0$ be such that $f(\omega_0;0)=0$  and let $\lambda_1(\omega_0), \dots, \lambda_m(\omega_0)$ be the $m$ nonzero eigenvalues of the generalised capacitance matrix $\Tilde{\mc C}(\omega_0)$. Then, as $\delta \to 0$, the $n=n(\omega_0)$ zeros of $f(\omega;\delta)$ near $\omega_0$ admit, after reordering, the following asymptotic expansions: 
    \begin{equation}
        \begin{cases}
            \omega^\pm_i(\delta) = \omega_0 \pm v \sqrt{\delta \lambda_i(\omega_0)}+o(\sqrt{\delta}), & i = 1,\dots, m,\\
            \omega_i(\delta) = \omega_0 + \mc{O}(\delta), & i = 2m+ 1,\dots, n.
        \end{cases}
    \end{equation}
\end{lemma}

Note that in the subwavelength regime, when we search for resonant frequencies that are in a neighborhood of zero (\emph{i.e.}, for $\omega_0=0$), the generalised capacitance matrix formalism reduces to the study of $\Tilde{\mc C}(0)$, which was first introduced in \cite{feppon.cheng.ea2023Subwavelength}. It is also worth noting that if the wave speeds $v_i$ are all real positive, then  $\Tilde{\mathcal{C}}(0)$ is always diagonalisable, whereas this is not necessarily the case for general $\omega_0 \in \{\omega\mid f(\omega;0)= 0\}$. The diagonalisability of $\Tilde{\mathcal{C}}(0)$ when the wave speeds are all positive follows from the fact that $V \Tilde{\mathcal{C}}(0)$,
where $V$ is the diagonal positive definite matrix given by $V:= \mathrm{diag}(t_1,t_3,\dots, t_{2N-1})$, is a real symmetric matrix and hence $\Tilde{\mathcal{C}}(0)$ is similar to $\sqrt{V^{-1}} V \Tilde{\mathcal{C}}(0) \sqrt{V^{-1}}$.

\section{Exceptional points}\label{sec:ep}

In this section, our aim is to elucidate the variety of exceptional points that arise in the present setting, both under leading-order (in $\delta$) non-Hermiticity caused by complex material parameters and weaker non-Hermiticity arising solely from radiation conditions. We begin by rigorously defining exceptional points for our problem and linking them to higher-order singularities of the Green's function of the resonator system, as well as to higher-order zeros of the characteristic determinant $f(\omega;\delta)$ in both the subwavelength and non-subwavelength regimes.

\subsection{Exceptional points and Green's function}
We first introduce the following definition of an exceptional point.
\begin{definition}[Exceptional point]
    Consider a finite array of resonators with a characteristic determinant $f(\omega;\delta)$. For some scattering resonance $\omega_0\in \C$, we define the \emph{algebraic multiplicity} $m_a(\omega_0)$ as the order of the zero $\omega_0$ of $f(\omega; \delta)$ and let the \emph{geometric multiplicity} $m_g(\omega_0)$ be the dimension of the space of solutions of \cref{eq:HH_coupled} when $\omega=\omega_0$. 
    We say that $\omega_0$ is an \emph{exceptional point} of  \cref{eq:HH_coupled}  if $m_a(\omega_0) > m_g(\omega_0)$.
\end{definition}
When $\delta > 0$, the propagation matrix enables the unique propagation of a solution through the resonator array, ensuring that $m_g(\omega_0) \leq 1$. Consequently, any higher-order zero of $f(\omega;\delta)$ corresponds to an exceptional point.

The key property of exceptional points is that degenerate resonances exhibit \emph{higher-order sensitivity} with respect to generic perturbations \cite[Chapter II, 1.2]{kato1995Perturbation}. In particular, for an exceptional point of order two, the two degenerate resonant frequencies $\omega_\pm$ perturb as
\[
    \omega_\pm = \omega_0 \pm \sqrt{\lambda\varepsilon}
\]
for the perturbation size $\varepsilon>0$ and some constant $\lambda>0$. 
As a consequence of this discussion, we therefore expect this \emph{square-root} sensitivity to system perturbation whenever two resonant frequencies of a system coincide.

We next demonstrate that the Green's function associated with problem \eqref{eq:HH_coupled} exhibits higher-order poles (order $> 1$) at exceptional resonant frequencies. The inhomogeneous problem corresponding to \eqref{eq:HH_coupled}, for some $f \in L^2(\R)$ supported on $(x_1^{\iL}, x_N^{\iR})$, is given by
\begin{equation} \label{eq_dimmer_inhomo_no_bc}
\begin{dcases}
    \frac{{\dd}^{2} u}{\dd x^2}(x) +\frac{\omega^2}{v_i^2}u(x)=f, & x\in D_i,\\
    \frac{\dd^2 u}{\dd x^2}(x)  + \frac{\omega^2}{v^2} u(x)=f, & x\in\R\setminus \mc D,\\
        u\vert_{\iR}(x^{\iLR}_i) - u\vert_{\iL}(x^{\iLR}_i) = 0, & \text{for all } 1\leq i\leq N,\\
        \left.\frac{\dd u}{\dd x}\right\vert_{\iR}(x^{\iL}_{{i}})=\delta\left.\frac{\dd u}{\dd x}\right\vert_{\iL}(x^{\iL}_{{i}}), & \text{for all } 1\leq i\leq N,\\
        \delta\left.\frac{\dd u}{\dd x}\right\vert_{\iR}(x^{\iR}_{{i}})=\left.\frac{\dd u}{\dd x}\right\vert_{\iL}(x^{\iR}_{{i}}), & \text{for all } 1\leq i\leq N,
    \end{dcases}
\end{equation}
with
\begin{equation} \label{eq_radiation_bc}
\big(u^{\prime}+\i \frac{\omega}{v} u\big) (x_1^{\iL})\big|_{\iL} = 0, \quad \big(u^{\prime}-\i \frac{\omega}{v} u\big) (x_N^{\iR})\big|_{\iR} = 0.
\end{equation}
Whenever \eqref{eq_dimmer_inhomo_no_bc} is well-posed, it is solved by
\begin{equation*}
u(x)=\int_{x_1^{\iL}}^{x_N^{\iR}}G(x,y;\omega)f(y)\dd y,
\end{equation*}
where the Green function is given by (cf. \cite[Sec.\,2.3, Chap.\,3]{kato1995Perturbation} and the references therein)
\begin{equation} \label{eq_dimer_green_func}
G(x,y;\omega)=\left\{
\begin{aligned}
&-\frac{u_1(x;\omega)u_2(y;\omega)}{W(y;\omega)},\quad x\leq y,\\
&-\frac{u_2(x;\omega)u_1(y;\omega)}{W(y;\omega)},\quad x\geq y.\\
\end{aligned}
\right.
\end{equation}
Here, $u_1(x;\omega)$ solves the homogeneous equation associated with \eqref{eq_dimmer_inhomo_no_bc} and \textit{the left boundary condition} $\big(u_1^{\prime}+\i \frac{\omega}{v} u_1\big) (x_1^{\iL};\omega)\big|_{\iL}=0$, $u_2$ solves the same equation but \textit{with a right boundary condition}, and $W$ is the Wronskian of $u_1$ and $u_2$. 

From \eqref{eq_dimer_green_func}, it follows that the singularity of the Green function is determined by the zeros of the Wronskian, as the divisor does not vanish identically as a function on $[x_1^{\iL}, x_N^{\iR}]$. Furthermore, since the Wronskian is either non-vanishing or identically vanishing, studying its zeros is equivalent to analysing the zeros of the boundary value $W(x_N^{\iR})$: 
\begin{equation*}
W(x_N^{\iR})=\det
\begin{pmatrix}
u_1(x_N^{\iR};\omega) & u_2(x_N^{\iR};\omega) \\
u_1^{\prime}(x_N^{\iR};\omega) & u_2^{\prime}(x_N^{\iR};\omega)
\end{pmatrix}
\propto\mathrm{det}\left(P_{tot}^{\omega}
    \begin{pmatrix}
        1 \\ -\i \frac{\omega}{v}
    \end{pmatrix}\middle|
    \begin{pmatrix}
        1 \\ \i \frac{\omega}{v}
    \end{pmatrix}
    \right).
\end{equation*}
Hence, we conclude that the following result holds. 
\begin{proposition}
The (highest) order of the poles of the Green function $G(x, y; \omega)$ is equal to the (highest) order of the zeros of $f(\omega; \delta)$. 
\end{proposition}

Note that $f(\omega;\delta)$ is an analytic function of $\omega$. Therefore, we arrive at the following characterisation of the highest order of the zeros of $f(\omega; \delta)$. 
\begin{proposition}
The highest order of zeros of $f(\omega;\delta)$ equals $n$ if and only if
\begin{equation*}
\sum_{i+j+k=m}\det\left(\big(\partial_{\omega}^{i}P_{tot}^{\omega}\big) 
\partial_{\omega}^{j} \begin{pmatrix}
    1 \\ -\i \frac{\omega}{v}
\end{pmatrix} \middle| 
\partial_{\omega}^{k}\begin{pmatrix}
    1 \\ \i \frac{\omega}{v}
\end{pmatrix}
\right)=0\quad \text{for any $m< n$}.
\end{equation*}
\end{proposition}

\subsection{Exceptional points in the limit $\delta=0$}
A subtle but key distinction is the differentiation between exceptional points that induce higher-order sensitivity to system parameter perturbations and exceptional points in the limit $\delta = 0$. At this limit, the system exhibits a variety of exceptional points leading to a slower convergence of the resonant frequencies as $\delta\to 0$. However, at $\delta=0$ these frequencies do not display higher-order dependence on the remaining system parameters, distinguishing them from typical exceptional points. 

\begin{figure}
    \centering
    \includegraphics[width=0.9\linewidth]{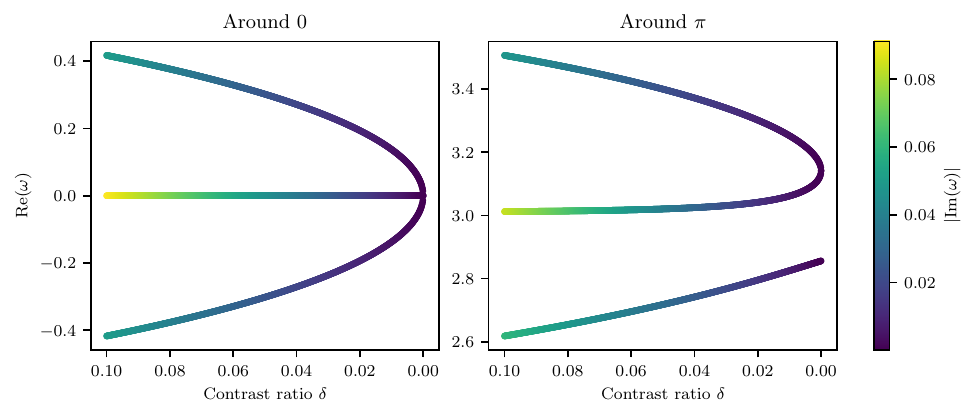}
    \caption{Convergence of resonant frequencies as $\delta\to 0$. We consider a dimer $N=2$ with equal wave speeds $v_1=v_2=v = s=\ell_1=1$ and $\ell_2=1.1$. The mismatch in resonator length leads to different asymptotic convergence rates around $0$ and $\pi$.}
    \label{fig:convergence_dimer}
\end{figure}

The characterisation in Lemma \ref{lem:deltazero_resonances} clearly indicates that, at $\delta = 0$, the resonant frequencies vary linearly or inversely linearly with respect to the system parameters $(\ell_j, v_j, s_j)$. However, \cref{thm:leadingorder} suggests that when resonant modes interact locally, they give rise to exceptional points at $\delta = 0$, resulting in slower (lower-order) convergence.

This behavior is illustrated in \cref{fig:convergence_dimer} for a perturbed dimer system with resonator lengths $\ell_1=1$ and $\ell_2 = 1.1$, focusing on resonances near $0$ and $\pi$. Around $0$, all eigenmodes interact, resulting in $\sqrt{\delta}$-order convergence. However, the situation is more nuanced around $\pi$: due to $s = \ell_1 = 1$, the resonant mode of the first resonator interacts with that of the gap between the resonators, again yielding $\sqrt{\delta}$-order convergence to the limiting mode. In contrast, the resonant mode of the second resonator does not interact with the eigenmode of the gap between the resonators due to the length mismatch, resulting in $\delta$-order convergence. 

\subsection{Approximate exceptional points from leading-order non-Hermiticity}\label{ssec:complex_ep}

In this subsection, we examine the exceptional points arising from leading-order non-Hermiticity, induced by complex material parameters, similar to the analysis in \cite{exceptpoint1}.

By \cref{thm:leadingorder}, the first-order asymptotics of the resonant frequencies are determined by the generalised capacitance matrix \eqref{eq:mathcalCk_def}, which as noted before is Hermitian at $\omega=0$ for systems with real system parameters. This leads to the following result.

\begin{proposition}\label{prop:no_ep}
    Consider a resonator array  with real system parameters $v_i, \delta \in \R$. Then there exists $\delta_0>0$ such that for all $0<\delta<\delta_0$, the characteristic determinant $f(\cdot; \delta)$ has only simple roots in a neighbourhood of $0$. In particular, such a system cannot exhibit any exceptional points in the subwavelength regime.
\end{proposition}

Therefore, exceptional points in the subwavelength regime must arise from leading-order non-Hermiticity introduced by complex system parameters. In particular, these exceptional points appear at the level of the capacitance matrix $\Tilde{\mc C}(0)$ and are therefore stable with respect to $\delta$.

\begin{figure}[h]
    \centering
    \begin{subfigure}[t]{0.4\textwidth}
        \centering
        \includegraphics[width=\textwidth]{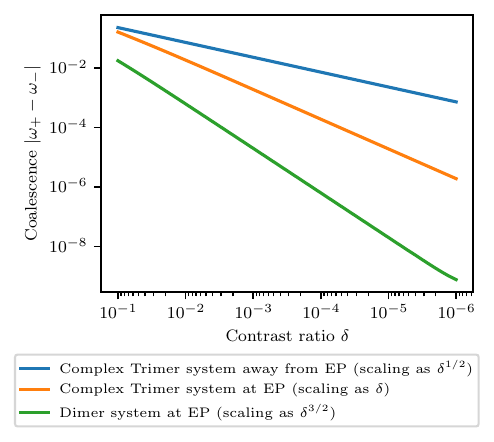}
        \caption{Rate of coalescence of $\omega_\pm$ for varying system configurations.}
        \label{sfig:coalescence_rate}
    \end{subfigure}%
    \hfill
    \begin{subfigure}[t]{0.5\textwidth}
        \centering
        \includegraphics[width=\textwidth]{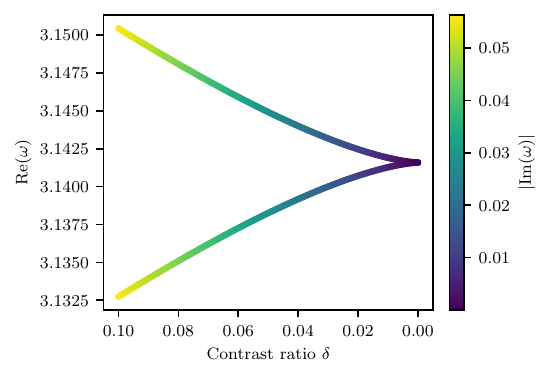}
        \caption{$\mc O(\delta^\frac{3}{2})$-order convergence of $\omega_\pm$ to $\pi$ for the dimer system.}
        \label{sfig:radiation_ep_coalescence}
    \end{subfigure}
    \caption{Coalescence of resonant frequencies as $\delta\to 0$. For the trimer system with complex material parameters from \cref{ssec:complex_ep}, $\omega_\pm$ denote the two resonances converging to $0$ with positive real part as $\delta\to 0$. We compare the cases at and away from the exceptional point by tuning $\theta$ to $\theta=\pi/4$ and $\theta=0$. 
    For the dimer system from \cref{ssec:rad_ep}, $\omega_\pm$ denote the two resonances converging to $\pi$. At exceptional points, we observe higher-order coalescence of the corresponding resonances.
    }
    \label{fig:coalescence}
\end{figure}

We investigate a trimer of resonators ($N=3$) with complex material parameters parametrised by $(v_1, v_2, v_3) = (e^{\i \theta}, 1, e^{-\i\theta})$ for $\theta\in [0,\pi]$ with $s_i=\ell_i=1$ for all $i=1,2,3$. The generalised capacitance matrix \eqref{eq:mathcalCk_def} at $\omega = 0$ is then given by 
$$
\Tilde{\mc C}(0) = \begin{pmatrix}
        e^{2\i \theta}&- e^{2\i \theta}&0\\
        -1 & 2 & -1\\
        0 & -  e^{-2\i\theta} & e^{-2\i\theta}
    \end{pmatrix},
$$
and its eigenvalues are 
\[
    (\lambda_1, \lambda_2, \lambda_3) = (0,1,e^{2\i\theta}+1+e^{-2\i\theta}).
\]
Therefore, $\Tilde{\mc C}(0)$ reaches its first exceptional point at $\theta_c = \pi/4$, where the eigenvalue $\lambda_2 = \lambda_3 = 1$ possesses an algebraic multiplicity of $2$ and a geometric multiplicity of $1$. 

\begin{remark}
    We have chosen a trimer of resonators with parity-time-symmetry ($v_1= \overline{v_3}; v_2 \in \mathbb{R}$) in order to obtain an exceptional point for the capacitance matrix  that is away from zero (which is not possible for dimer systems).  Exceptional points of a capacitance matrix that are at zero require special considerations since they lead to a third-order degeneracy of the characteristic determinant $f(\omega; \delta)$ at $\delta=0$, resulting in asymptotic cube-root behaviour.
\end{remark}

As the contrast $\delta \to 0$, \cref{thm:leadingorder} indicates that the subwavelength resonances $\omega_i$ in the positive half-plane are, in leading order, determined by the eigenvalues \(\lambda_i\), \emph{i.e.},  
\[
    \omega_i = \sqrt{\delta\lambda_i} + \mc {O}(\delta).
\]
Therefore, we expect that the resonant frequencies associated with $\lambda_2$ and $\lambda_3$ will go to $0$ with order $\mc{O}(\sqrt{\delta})$, but they will coalesce with higher order at the critical value $\theta_c = \pi/4$:
\[
    \omega_3-\omega_2 = \begin{cases}
        \mc O(\delta) & \text{if } \theta=\pi/4, \\
        \mc O(\sqrt{\delta}) & \text{else}.
    \end{cases}
\]
Here, $\omega_3,\omega_2$ denote the two subwavelength resonant frequencies with positive real parts as $\delta\to 0$. In fact, this coalescence dichotomy is observed in \cref{fig:coalescence}(\textsc{a}) for the trimer system at and away from the exceptional point after tuning $\theta$ (in particular, we compare the coalescence behaviour for  $\theta=\pi/4$ and $\theta=0$).

For small $\delta>0$, the higher order coalescence of the resonant frequencies acts as an approximate exceptional point. In particular, the eigenvalues of $\Tilde{\mc C}(0)$ and thus the leading-order behaviour of $\omega_1, \omega_2$ display higher-order sensitivity to system parameters. This is demonstrated in \cref{fig:ep_stability}(\textsc{a}) where, as the first resonator size $\ell_1$ is varied around $1$, the resonant frequencies $\omega_1, \omega_2$ demonstrate the characteristic higher-order sensitivity to perturbations. At this point, we would also like to note that, while this constitutes only an approximate exceptional point, the discussion from the following subsection would suggest that there is an exact exceptional point in a system-parameter neighbourhood.

Furthermore, it is worth emphasising that while this construction yields approximate exceptional points in the subwavelength regime as $\delta\to 0$, it is unable to do so at higher scattering resonances. In particular, \cref{thm:leadingorder} implies that two resonators $D_i$ and $D_j$ can only interact at higher resonances if $\ell_i v_i \equiv \ell_j v_j\; [\pi]$ . However, because the above construction requires that the wave speeds are complex conjugate to each other, this implies that, as $\delta\to 0$, resonators with distinct complex wave speeds can maintain their interaction only in the subwavelength regime.

\subsection{Approximate exceptional points from radiation losses}\label{ssec:rad_ep}

In the subwavelength regime, \cref{prop:no_ep} prevents the possibility of obtaining exceptional points purely from non-Hermiticity due to radiation losses. However, the non-subwavelength regime offers the possibility of resonator interactions of the same order as the radiation losses as $\delta\to 0$. We therefore aim to construct a system in which the resonators interact weakly. To do so, we consider a dimer system ($N=2$) and let $s=1/2$ and $\ell_2=v_2=1$. 
To ensure that the isolated frequency of the first resonator matches that of the second one, we then choose the ansatz $\ell_1 = v_1 = \theta$ and aim to tune $\theta(\delta)$ until the system passes through an exceptional point. 

\begin{figure}
    \centering
    \includegraphics[width=1\textwidth]{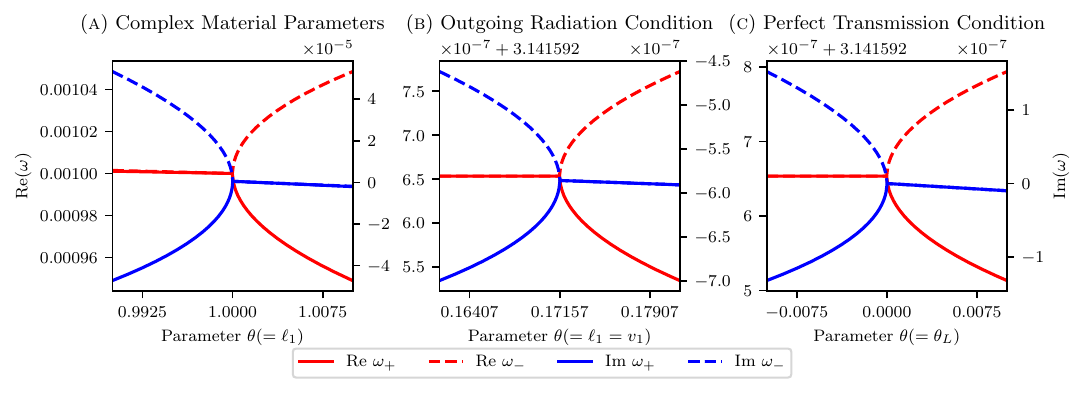}
    \caption{Characteristic square-root sensitivity to parameter perturbations around the approximate and exact exceptional points for the systems from \cref{ssec:complex_ep,ssec:rad_ep,ssec:prc_ep}. For all systems, we choose $\delta=10^{-6}$ and observe that the exceptional points closely match the square-root sensitivity profile with respect to $\theta$. 
    }
    \label{fig:ep_stability}
\end{figure}

Choosing $s/v \neq \ell_1/v_1=\ell_2/v_2=1$
plays the crucial role of limiting the interaction between the two resonators. In this case, $\Tilde{\mc C}(\omega_0 = \pi)$ is the zero matrix and \cref{thm:leadingorder} leads to the leading order ansatz $\omega_\pm = \pi + \lambda_\pm\delta + \mc O(\delta^{\frac{3}{2}})$ for the two resonant frequencies converging to $\pi$ as $\delta\to 0$. Here, $\lambda_\pm$ are the leading order parameters. Plugging this into $f(\omega;\delta)$ yields
\[
    f(\pi + \lambda\delta; \delta) = -2\pi + \i\pi (1+\frac{1}{\ell_1})\lambda + \frac{\pi}{\ell_1}\lambda^2 + \mc O(\delta).
\]
The solutions of this polynomial of degree two yield the leading order behaviour $\lambda_\pm$ of the two resonant frequencies approaching $\pi$, and setting its discriminant to zero yields the condition for an asymptotic exceptional point
\[
   -1+6\theta-\theta^2 = 0,
\]
determining the critical values $\theta_c = 3\pm 2\sqrt{2}$.
At these critical values, the ansatz $\omega_\pm = \pi + \lambda_\pm\delta + \mc O(\delta^{\frac{3}{2}})$ suggests a coalescence of $\omega_\pm$ with order $\mc O(\delta^{3/2})$, yielding an approximate exceptional point. This coalescence behaviour is demonstrated in \cref{fig:coalescence}. In \cref{fig:ep_stability}(\textsc{b}), we observe that the approximate exceptional point indeed also displays the expected square-root sensitivity to perturbations in $\theta$ around the critical values.

\begin{remark}
    The $\mc O(\delta^{3/2})$-order coalescence of $\omega_\pm$ suggests that the system rapidly achieves a high-quality approximate exceptional point as $\delta\to0$. By expanding $\omega_\pm$ to higher orders of $\delta$ and possibly varying the system parameters with $\delta$, the coalescence speed and thus the exceptional point quality can be improved even further.
\end{remark}

\subsection{Exact exceptional points from parity-time symmetric radiation conditions}\label{ssec:prc_ep}

In the previous two subsections, approximate exceptional points were achieved by ensuring the coalescence of the leading-order behavior of resonances as $\delta \to 0$. However, while \cref{eq:HH_coupled} inherently includes a source of non-Hermiticity via the outgoing radiation conditions, the energy gain and loss imposed by these conditions are \enquote{unbalanced}, making it challenging to realise exact exceptional points explicitly. 

\begin{figure}[ht]
    \centering
    \includegraphics[width=0.8\linewidth]{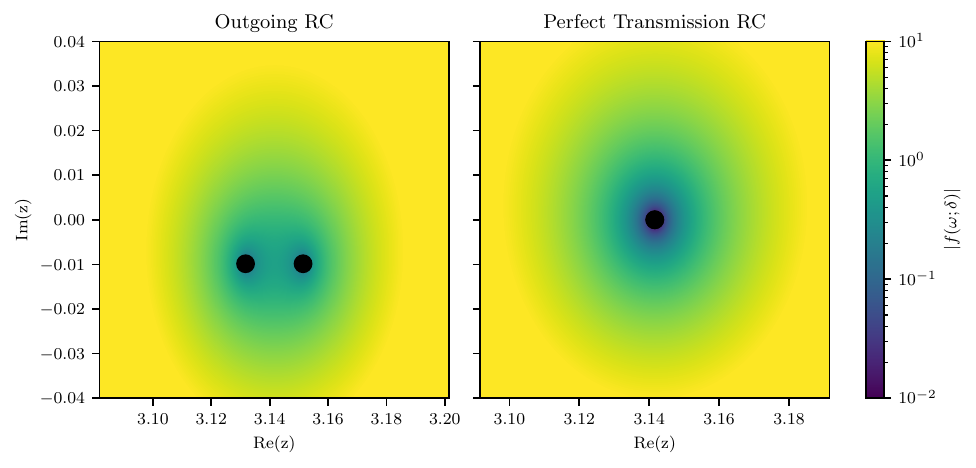}
    \caption{Resonances, marked by black dots, as characterised by solutions of $f(\omega;\delta)=0$ for the standard outgoing radiation conditions and solutions of $\Tilde{f}(\omega;\delta)=0$ for the perfect transmission radiation conditions. We consider the modified dimer with $N=2, \ell_i = 1, v_i=1$ and $s_1 = 3/2$ and $\delta=10^{-2}$.}
    \label{fig:RC_comparison}
\end{figure}

On the operator level, the desired type of \enquote{balance} is expressed as parity-time ($\mc{PT}$) symmetry condition.

\begin{definition}[Parity-time symmetry]
  For functions $u \in L^2(\R)$, we define the \emph{parity operator} $\mc{P}[u](x) := u(-x)$ and the \emph{time-reversal operator} $\mc{T}[u](x) := \overline{u(x)}$, where $\overline{\cdot}$ denotes complex conjugation. An operator $L$ on $L^2(\R)$ is said to be \emph{$\mc{PT}$-symmetric} if it satisfies $[L, \mc{PT}] = 0$.
\end{definition}
An immediate consequence of $\mc{PT}$-symmetry is that the resonant frequencies of $L$ are symmetric with respect to the real axis in the complex plane. This is significant because it compels resonant frequencies to coalesce on the real line, leading to the formation of exceptional points.

To achieve $\mc{PT}$-symmetry at the differential equation level, we choose a spatially reflection symmetric resonator array together with a modified radiation condition known as the \emph{perfect transmission condition}; see \cite{fleury.sounas.ea2015invisible, PhysRevResearch.3.013223}. The modified boundary conditions are:
\begin{equation} \label{eq:RadiationBCPerfect}
\frac{\dd u}{\dd x}(x_1^{\iL}) - \i \frac{\omega}{v} u(x_1^{\iL}) = 0, \quad \frac{\dd u}{\dd x}(x_N^{\iR}) -\i \frac{\omega}{v} u(x_N^{\iR}) = 0.
\end{equation}

Therefore, instead of \eqref{eq:chdet}, the scattering resonance problem corresponds to investigating the zeros of the following analytic function: 
\[
\Tilde{f}(\omega;\delta):= \det \left(P_{tot}^{\omega}
\begin{pmatrix}
    1 \\ \i \frac{\omega}{v}
\end{pmatrix}\middle|
\begin{pmatrix}
    1 \\ \i \frac{\omega}{v}
\end{pmatrix}
\right).
\]
The flipped sign of the radiation condition at the left edge ensures only incoming radiation at this edge, which enforces perfect transmission. We now consider a symmetric dimer system with $N=2, \ell_i = 1, v=v_i=1$, parametrised by the spacing $s$. The resulting differential equation model is $\mc{PT}$-symmetric and we explicitly find that  
\[
    \Tilde{f}(\omega;\delta) = \omega\frac{\delta^2-1}{\delta^2} \sin(\omega) \left[ -2 \delta \cos(\omega) \cos\left(\omega s\right) + \left(1 + \delta^2\right) \sin(\omega) \sin\left(\omega s\right) \right].
\]
This immediately shows that the system has a resonant frequency $\omega = \pi$ for any $\delta > 0$ and $s > 0$. However, when the spacing $s$ is chosen as one of the critical values in $\frac{1}{2} + \Z_{\geq 0}$, the zero becomes of order two, resulting in an exact and explicit exceptional point for any $\delta > 0$.

In \cref{fig:RC_comparison}, we show the two closest resonances to $\omega=\pi$ of the dimer system, both under standard outgoing radiation conditions \eqref{eq:RadiationBC} and under perfect transmission radiation conditions \eqref{eq:RadiationBCPerfect}, demonstrating the role the radiation conditions play in facilitating this exceptional point. 

\cref{fig:ep_stability}(\textsc{c}) demonstrates square-root perturbation of the resonant frequencies as the left boundary conditions are varied around the perfect transmission condition. The left boundary condition of the system is parametrised by $\theta_L \in [0,2\pi)$ and enforced by requiring,
\[
    \begin{pmatrix}
        u(x_1^{\iL})\\
        \frac{\dd u}{\dd x}\rvert_{\iL}(x_1^{\iL})
    \end{pmatrix} \propto \cos \theta_L \begin{pmatrix}
        1\\
        \i\omega
    \end{pmatrix} + \sin \theta_L \begin{pmatrix}
        1\\
        -\i\omega
    \end{pmatrix}.
\]
\section{Non-reciprocal systems} \label{sec:NR}
We now turn our attention to Fabry-Pérot resonances in non-reciprocal systems. To that end, we modify \cref{eq:HH} by introducing a non-reciprocal imaginary gauge potential (\cite{ammariMathematicalFoundationsNonHermitian2024,yokomizo}) and obtain the following generalised Strum-Liouville equation:
\begin{align}
    -\frac{\omega^{2}}{\kappa(x)}u(x)- 2\gamma(x){\frac{1}{\rho(x)}} \frac{\dd u}{\dd x}(x)-\frac{\dd}{\dd x}\left( \frac{1}{\rho(x)}\frac{\dd u}{\dd
    x}  (x)\right) =0,\qquad x \in\R,
    \label{eq: gen Sturm-Liouville}
\end{align}
again subject to the outgoing radiation conditions
\begin{align}
\label{radiation}
\frac{\dd u}{\dd \abs{x}}(x) -\i \frac{\omega}{v} u(x) = 0, \qquad  x\in(-\infty,x_1^{\iL})\cup (x_N^{\iR},\infty).
\end{align}
Here, we introduce the new parameter $\gamma$ in the form of the \emph{imaginary gauge potential}, which is also assumed to be piecewise constant,
\begin{align*}
\gamma(x) = \begin{dcases}
    \gamma_i,\quad x\in D_i,\\
    0,\quad x \in \R\setminus \mc D.
\end{dcases}
\end{align*}

Similarly to \eqref{eq:HH_coupled}, since the material parameters are piecewise constants, \eqref{eq: gen Sturm-Liouville} together with \eqref{radiation} can be written again as 
\begin{align}\label{eq:skin_coupled}
    \begin{dcases}
    \frac{{\dd}^{2} u}{\dd x^2}(x) + 2\gamma_i \frac{\dd u}{\dd x}(x) +\frac{\omega^2}{v_i^2}u(x)=0, & x\in D_i,\\
    \frac{\dd^2 u}{\dd x^2}(x)  + \frac{\omega^2}{v^2}u(x)=0, & x\in\R\setminus \mc D,\\
        u\vert_{\iR}(x^{\iLR}_i) - u\vert_{\iL}(x^{\iLR}_i) = 0, & \text{for all } 1\leq i\leq N,\\
        \left.\frac{\dd u}{\dd x}\right\vert_{\iR}(x^{\iL}_{{i}})=\delta\left.\frac{\dd u}{\dd x}\right\vert_{\iL}(x^{\iL}_{{i}}), & \text{for all } 1\leq i\leq N,\\
        \delta\left.\frac{\dd u}{\dd x}\right\vert_{\iR}(x^{\iR}_{{i}})=\left.\frac{\dd u}{\dd x}\right\vert_{\iL}(x^{\iR}_{{i}}), & \text{for all } 1\leq i\leq N,\\
        \frac{\dd u}{\dd \abs{x}}(x) -\i \frac{\omega}{v} u(x) = 0, & x\in(-\infty,x_1^{\iL})\cup (x_N^{\iR},\infty).
    \end{dcases}
\end{align}
In addition, we choose $L > 0$ minimal such that $\mc D \subset [0,L]$, assuming without loss of generality that $x_1^{\iL}=0$. In contrast to \cref{eq:HH_coupled} from the previous sections, we assume that $\rho(x), \kappa(x)$ are real and positive for the remainder of this section.  We further assume that $v=v_i=1$.
 
\subsection{Propagation matrix approach}\label{ssec:skin_propmat}

In this section, our aim is to understand the non-reciprocal scattering resonance problem \cref{eq:skin_coupled} through its propagation matrix.
As in \cref{psetting}, to obtain the total propagation matrix associated with \eqref{eq:skin_coupled}, we first look for the propagation matrix of a single block, that is, an analogue of \eqref{pomegai}. Comparing \eqref{eq:skin_coupled} with \eqref{eq:HH_coupled}, we see that the gauge potential is applied only within the resonator; thus, only the interior propagation matrix $P^\omega_{i,int}$ requires modification. The modified interior propagation matrix is given as follows.
\begin{lemma} 
Let $P^\omega_{i,\mathrm{int},\mathrm{gauge}}$ denote the modified interior propagation matrix for problem~\eqref{eq:skin_coupled}, defined analogously to~\eqref{def:propga}. 
We have
\begin{equation*}
P^\omega_{i,int,gauge}=e^{-\gamma_i\ell_i}\begin{pmatrix}
1 & 0 \\
-\gamma_i & 1
\end{pmatrix}P^{\sqrt{\omega^2-\gamma_i^2}}_{i,int}
\begin{pmatrix}
1 & 0 \\
\gamma_i & 1
\end{pmatrix} .
\end{equation*}
\end{lemma}
\begin{proof}
Applying the gauge transform \begin{equation} \label{symm} v(x):=e^{\gamma_i x}u(x) \end{equation} and using \eqref{eq:skin_coupled}, it is easy to check that $v(x)$ satisfies
\begin{equation*}
\frac{{\dd}^{2} v}{\dd x^2}(x)  +(\omega^2-\gamma_i^2)v(x)=0,\quad  x\in D_i .
\end{equation*}
Hence, the propagation of $v$ is determined by $P^{\sqrt{\omega^2-\gamma_i^2}}_{i,int}$. The proof is then complete by incorporating the gauge factor into the Dirichlet and Neumann data.
\end{proof}
With this modified interior propagation matrix, we obtain the propagation matrix for the $i$\textsuperscript{th} block.
\begin{lemma}\label{lem:propagation}
    The $i$\textsuperscript{th} propagation matrix is given by 
    \begin{equation}
    \begin{aligned}
    P^\omega_{i,gauge} &= P^\omega_{i,ext}P_{ext\to int}^{-1}P^\omega_{i,int,gauge}P_{ext\to int} \\
    &=e^{-\ell_i\gamma_i}{\footnotesize\begin{pmatrix}
            \cos(\omega s_i)\cosh(\nu\ell_i)-\Psi(-\delta\gamma_i,\omega)\sinh(\nu\ell_i) &
            \frac{1}{\omega}\cosh(\nu\ell_i)\sin(\omega s_i) - \frac{\delta}{\omega}\Psi(-\omega,\frac{\gamma_i}{\delta})\sinh(\nu \ell_i) \\
            -\omega\cosh(\nu \ell_i)\sin(\omega s_i) -\frac{\omega}{\delta}\,\Psi(\omega,\delta\gamma_i)\sinh(\nu \ell_i) &
            \cos(\omega s_i)\cosh(\nu \ell_i)
            -\delta\Psi(\frac{\gamma_i}{\delta},\omega)\sinh(\nu \ell_i)
        \end{pmatrix}},
    \end{aligned}
    \end{equation}
where $\Psi(a,b) \coloneqq ({a\cos(\omega s_i)+b\sin(\omega s_i)})/{\nu}$ and $\nu \coloneqq \sqrt{\gamma_i^2-\omega^2}$. It satisfies the relation
\begin{equation*}
\begin{pmatrix}
        u(x_{i+1}^{\iL}) \big|_{\iL} \\
        u'(x_{i+1}^{\iL}) \big|_{\iL}
    \end{pmatrix} = P_{i,gauge}^\omega \begin{pmatrix}
        u(x_{i}^{\iL}) \big|_{\iL} \\
        u'(x_{i}^{\iL}) \big|_{\iL}
    \end{pmatrix},
\end{equation*}
for any $u(x)$ solution of \cref{eq:skin_coupled}. Furthermore, we have $\det(P^\omega_{i,gauge} ) = e^{-2\ell_i\gamma_i}$.
\end{lemma}
In the sequel,  we drop for simplicity the subscript `gauge' from $P_{i, gauge}^{\omega}$ as the gauge potential is always in place. To recast $P^\omega_{i}$ in a symmetric form, we factor out the decay and introduce the following \emph{symmetrised propagation matrix}: 
\begin{equation}\label{eq:symm_prop}
    \Tilde{P}_i^\omega \coloneqq e^{\ell_i \gamma_i}P_i^\omega \in \operatorname{SL}(2,\C).
\end{equation}

Now, we aim to study the resonant frequencies of \cref{eq:skin_coupled} and define analogously to \eqref{ptotal} the \emph{total propagation matrix} $P^\omega_{tot} \coloneqq P_{N}^\omega\cdots P_1^\omega$ and its symmetrisation:
\begin{align} \label{eq:symm_prop2}
\Tilde{P}^\omega_{tot}:= e^{\sum_{i=1}^N\ell_i\gamma_i} P^\omega_{tot}.
\end{align}
Here, we have again chosen $s_N=0$ in order to define $P_N^\omega$.
Analogously to \eqref{eq:characterisationExp}, by matching the radiation conditions \eqref{radiation}, we find that 
\begin{equation} \label{eq:characterisation}
    \omega \text{ is a resonant frequency of }\cref{eq:skin_coupled} \iff 
\det\left(P_{tot}^{\omega}
\begin{pmatrix}
    1 \\ -\i \omega
\end{pmatrix}\middle|
\begin{pmatrix}
    1 \\ \i \omega
\end{pmatrix}
\right) =0.
\end{equation}
Thanks to the symmetrised  $\Tilde{P}^\omega_{tot}$, we may also check the following (mathematically) equivalent characterisation:
{\begin{equation}\label{eq:ptilde_eva_criterion}
  \det\left(\Tilde{P}_{tot}^{\omega}
\begin{pmatrix}
    1 \\ -\i \omega
\end{pmatrix}\middle|
\begin{pmatrix}
    1 \\ \i \omega
\end{pmatrix}
\right) =0, 
\end{equation}
} 
which is significantly more stable numerically than \cref{eq:characterisation}. 
 

Note that in the subwavelength regime, the characterisation  \eqref{eq:characterisation} 
reduces to finding the eigenvalues of the (non-symmetric) gauge capacitance matrix $\mc C^\gamma$ that is given by
\begin{equation} \label{gaugecapacitance}
\mc C^\gamma :=  \begin{dcases}
            \frac{\zeta(\gamma_1\ell_1)}{s_1},                    & i=j=1,                 \\
             \frac{\zeta(-\gamma_i\ell_i)}{s_{i-1}}+\frac{\zeta(\gamma_i\ell_i)}{s_i},                & 1< i=j< N,             \\
            \frac{\zeta(-\gamma_N\ell_N)}{ s_{N-1}},       & i=j=N,                 \\
            -\frac{\zeta(\pm\gamma_i\ell_i)}{s_{\min(i,j)}} , & 1\leq i\pm 1=j \leq N, \\
            \ 0,                                                 & \text{else},
        \end{dcases}
\end{equation}
where \begin{equation*} \label{def:zeta}
\zeta(z)\coloneqq \frac{z}{1-e^{-z}}>0;
\end{equation*}
see \cite[Section 5]{pm1}.
Moreover, by factoring out the non-reciprocal decay as done in the characterisation \eqref{eq:ptilde_eva_criterion}, we can find a similar but symmetric
system to $\mc C^\gamma$; see \cite{ammari2025competingedgebulklocalisation}.

We also note that, analogously to \cref{thm:leadingorder}, a characterisation of higher resonant frequencies for non-reciprocal systems can be derived using the same approach as in \cite{pm2}. Nevertheless, we do not make use of it in the sequel.  

\subsection{Symmetrisation by gauge transformation} \label{sec:symgauge}

Intuitively, transitioning from \cref{eq:characterisation} to \cref{eq:ptilde_eva_criterion} in \cref{ssec:skin_propmat} reformulates the original non-reciprocal problem into a better-conditioned reciprocal one by factoring out the decay. In this section, we rigorously establish this from an operator perspective by transforming \eqref{eq: gen Sturm-Liouville} into a self-adjoint problem through a gauge transformation.

The appropriate Hilbert space to study the present Sturm-Liouville problems is $\hs$ with weight $\dd\mu = \frac{1}{\kappa(x)}\dd x$, yielding the inner product
\begin{equation}
    \ip{f}{g}_{\hs} \coloneqq \int_{-\infty}^\infty f(x) \overline{g(x)} \frac{1}{\kappa(x)} \dd x.
\end{equation}
The problem \eqref{eq: gen Sturm-Liouville} then corresponds to the spectral problem $\mc L u = \omega^2 u$ of the non-Hermitian operator $\mathcal{L}: \dom(\mathcal{L}) \subset \hs \to \hs$ that is defined by
\begin{equation}
    \mathcal{L} u \coloneqq -\kappa(x) \left[ \frac{\dd}{\dd x}\left( \frac{1}{\rho(x)}\frac{\dd u}{\dd x} \right) + 2\frac{\gamma(x)}{\rho(x)} \frac{\dd u}{\dd x} \right].
\end{equation}
The domain $\dom(\mathcal{L})$ consists of functions $u \in \hs$ such that $u$ and $\frac{1}{\rho}u'$ are locally absolutely continuous, ensuring that the transmission conditions in \eqref{eq:skin_coupled} are well-defined, and that $\mathcal{L}u \in \hs$. 

The crucial step in the symmetrisation of $\mc L$ is to introduce the \emph{gauge transformation} $T: \hs \to \hs$ as the multiplication operator:
\begin{equation} \label{def:gaugetrnas}
    T [v](x) = e^{-\Gamma(x)} v(x)\,, \quad \text{where}\quad \Gamma(x) \coloneqq \int_0^x\gamma(x')\dd x'\,.
\end{equation}

\begin{proposition}[Similarity]\label{prop:gauge_similarity}
   The \emph{gauge transform} $T$ in \eqref{def:gaugetrnas} is a bounded linear isomorphism with a bounded inverse. Furthermore, the operator $\mathcal{L}$ is similar to the operator $\mc S$, that is, $\mathcal{L} = T\mc S T^{-1}$, where $\mc S$ is defined by
    \begin{equation}
        \label{eq:S_def}
        \mc S v \coloneqq -\kappa(x) \frac{\dd}{\dd x} \left( \frac{1}{\rho(x)} \frac{\dd v}{\dd x} \right) + V(x) v
    \end{equation}
    with the potential $V$ being defined by $V(x) := \kappa(x) \left[ \frac{\gamma(x)^2}{\rho(x)} - \frac{\dd}{\dd x}\left(\frac{\gamma(x)}{\rho(x)}\right) \right]$.
\end{proposition}
\begin{proof}
    Since $\Gamma$ is bounded, there exist constants $c, C > 0$ such that $c \leq |e^{\pm \Gamma(x)}| \leq C$. Thus, $T$ and $T^{-1}$ are bounded operators on $\hs$.
    We determine $\mc S$ by computing $T^{-1} \mathcal{L} T v$.  Using the identity
    \[
    e^{\Gamma} \left( \frac{\dd}{\dd x} \right) e^{-\Gamma} = \frac{\dd}{\dd x} - \gamma
    \]
    together with the product rule, a direct calculation shows that the first-order derivative terms cancel, yielding the Sturm-Liouville form in \cref{eq:S_def}. Note that the term $\frac{\dd}{\dd x}(\frac{\gamma}{\rho})$ in $V$ must be understood in the distributional sense, resulting in Dirac delta potentials at the interfaces where material parameters jump.
\end{proof}

\begin{proposition}[Self-adjointness]
    The operator $\mc S$ defined in \eqref{eq:S_def}, equipped with the domain $\dom(\mc S) = \{v \in \hs \mid v, \frac{1}{\rho}(v'-\gamma v) \in AC_{loc}(\R), S v \in \hs \}$, is self-adjoint on the Hilbert space $\hs$. 
    Here, $AC_{loc}(\R)$ denotes the set of functions that are absolutely continuous on any compact subset of $\R$.
\end{proposition}

\begin{proof}
    Following standard Sturm-Liouville theory \cite[Chapter 3]{zettl2005SturmLiouville}, since the potential $V$ is real-valued, formal symmetry is sufficient to ensure self-adjointness of $\mc S$. Let $f, g \in \dom(S)$. We compute
    \begin{align*}
        \ip{\mc S f}{g}_{\hs} &= \int_{-\infty}^\infty \left[ -\kappa \frac{\dd}{\dd x} \left( \frac{1}{\rho} \frac{\dd f}{\dd x} \right) + V f \right] \overline{g} \, \frac{1}{\kappa} \dd x \\
        &= \int_{-\infty}^\infty \left[ - \frac{\dd}{\dd x} \left( \frac{1}{\rho} \frac{\dd f}{\dd x} \right) + \frac{V}{\kappa} f \right] \overline{g} \dd x.
    \end{align*}
    Integrating by parts (and noting that the boundary terms vanish for $L^2$ functions), we find that
    \begin{align*}
      \ip{\mc S f}{g}_{\hs}  &= \int_{-\infty}^\infty \frac{1}{\rho} \frac{\dd f}{\dd x} \frac{\dd \overline{g}}{\dd x} \dd x + \int_{-\infty}^\infty \frac{V}{\kappa} f \overline{g} \dd x \\
        &= \ip{f}{S g}_{\hs}.
    \end{align*}
    That is, $\mc S$ is symmetric.
\end{proof}

Moreover, the \emph{symmetrised total propagation matrix} $\Tilde{P}^\omega_{tot}$ from \cref{eq:symm_prop2} is precisely the propagation matrix of the symmetrised $\mc{S}$. To establish this, it suffices to relate $P^\omega_i$ from \cref{eq:symm_prop} to the operator $\mc{S}$, as discussed in the following lemma.

\begin{lemma}\label{lem:symm_propmat_is_propmat_of_symm}
    Let $[0,L]\subset U$ an open neighbourhood and  $v\in L^2(U, \dd \mu)$ be a solution to the resonance problem $(\mc S -\omega^2)v = 0$ on $U$.
    Then we have
    \begin{equation*}
        \begin{pmatrix}
        v(x_{i+1}^{\iL}\big|_{\iL})\\
        v'(x_{i+1}^{\iL})\big|_{\iL}
    \end{pmatrix} = \Tilde{P}_i^\omega \begin{pmatrix}
        v(x_{i}^{\iL})\big|_{\iL}\\
        v'(x_{i}^{\iL})\big|_{\iL}
    \end{pmatrix},
    \end{equation*}
    where $\Tilde{P}_i^\omega$ is the \emph{symmetrised propagation matrix} from \cref{eq:symm_prop}.
\end{lemma}
\begin{proof}
    Let $v\in L^2(U, \dd \mu)$ with $(\mc S -\omega^2)v = 0$. 
    By \cref{prop:gauge_similarity} this implies that the transformed $u = Tv = e^{-\Gamma}v$ satisfies $(\mc L -\omega^2)v = 0$.  We drop the subscript $|_{\iL}$ for simplicity. By \cref{lem:propagation}, we thus have
    \[
        \begin{pmatrix}
            u(x_{i+1}^{\iL})\big|_{\iL}  \\
            u'(x_{i+1}^{\iL}) \big|_{\iL}
        \end{pmatrix} = P_i^\omega \begin{pmatrix}
            u(x_{i}^{\iL}) \big|_{\iL} \\
            u'(x_{i}^{\iL}) \big|_{\iL}
        \end{pmatrix}.
    \]
Moreover, since the limits $u'(x_{i}^{\iL}) \big|_{\iL}$ and $v'(x_{i}^{\iL}) \big|_{\iL}$ are taken from the exterior of the resonators, where $\frac{\dd \Gamma}{\dd x}(x) = 0$, we obtain the following relation:
\[
        \begin{pmatrix}
            u(x_{i}^{\iL}) \big|_{\iL} \\
            u'(x_{i}^{\iL}) \big|_{\iL}
        \end{pmatrix}= e^{-\Gamma(x_i^{\iL})}\begin{pmatrix}
            v(x_{i}^{\iL}) \big|_{\iL} \\
            v'(x_{i}^{\iL}) \big|_{\iL}
        \end{pmatrix}.
    \]
    Combining this with the above propagation equation now yields
    \[
        \begin{pmatrix}
            v(x_{i+1}^{\iL}) \big|_{\iL} \\
            v'(x_{i+1}^{\iL}) \big|_{\iL}
        \end{pmatrix} = e^{\Gamma(x_{i+1}^{\iL})-\Gamma(x_i^{\iL})}P_i^\omega \begin{pmatrix}
            v(x_{i}^{\iL}) \big|_{\iL} \\
            v'(x_{i}^{\iL}) \big|_{\iL}
        \end{pmatrix},
    \]
    and the result follows from the fact that $e^{\Gamma(x_{i+1}^{\iL})-\Gamma(x_i^{\iL})}$ = $e^{\ell_i\gamma_i}$.
\end{proof}

\subsection{Non-Hermitian skin effect}
The results from the previous subsection will now enable us to prove the existence of the \emph{non-Hermitian skin effect} for finite periodic arrays, \emph{i.e.}, the exponential localisation of eigenmodes with resonant frequency in the band. Intuitively, the Floquet-Bloch theory guarantees that the resonant modes of $\mc S$ grow at most polynomially if $\omega$ approaches the band of the corresponding infinite periodic system. This polynomial growth is overpowered by the exponential localisation induced by the gauge transformation $T$, leading to exponential edge-localisation.

We consider a unit cell composed of $N$ resonators $\mc D$ supported on the interval $[0,L]$ and assume without loss of generality that $x_1^{\iL}=0$ and $x_N^{\iR}+s_N = L$, ensuring that there is a nonzero spacing $s_N$ between the final resonator and the right unit cell edge. This determines the material parameters $\rho(x), \kappa(x),$ and $ \gamma(x)$ which are constant outside of $[0,L]$.

We construct finite periodic arrays by replicating the unit cell $M$ times to the right, yielding the modified material parameters:  
\[
\rho_M(x), \kappa_M(x), \gamma_M(x) := \rho(x - mL), \kappa(x - mL), \gamma(x - mL),
\]
for $x \in [mL, (m+1)L]$ and $0 \leq m < M$, with the parameters extended as constants outside $[0, ML]$. Taking the limit as $M \to \infty$ and allowing $m \in \Z$, we define the infinite periodic limit as follows: for $x \in [mL, (m+1)L]$ and $m \in \Z$,  
\[
\rho_\infty(x):= \rho(x - mL), \quad \kappa_\infty(x) := \kappa(x - mL), \quad \gamma_\infty(x) := \gamma(x - mL).
\]
From the discussion in \cref{sec:symgauge}, we can construct the associated self-adjoint operators $\mc{S}_M$ and $\mc{S}_\infty$, where we additionally equip $\mc S_M$ with the outgoing radiation condition
\[
\frac{\dd v}{\dd \abs{x}}(x) -\i \frac{\omega}{v} v(x) = 0, \quad  x\in \R \setminus [0,ML].
\]

We define the \emph{cell propagation matrix} as the product of the symmetrised propagation matrices of its constituent resonators
\[
    \Tilde{P}(\omega) \coloneqq \Tilde{P}^\omega_N\cdots\Tilde{P}^\omega_1.
\]
Note that in this case there exists a nonzero final spacing $s_N>0$, making $\Tilde{P}^\omega_N$ well-defined.
The eigenvalues of $\Tilde{P}(\omega)\in \SL(2,\C)$ are $e^{\pm\i k(\omega)}$ with the convention $\Re k(\omega)\geq 0$. $k(\omega)$ is called the \emph{quasi-momentum} and controls the spatial behaviour of resonant solutions. 

We recall the following two results from the Floquet-Bloch theory of periodic operators; see \cite{eastham1973spectral}.
\begin{lemma}\label{lem:inspec_iff_inband} Let $\sigma(\mc S_\infty)$ denote the spectrum of $\mc S_\infty$. The following equivalence holds:
    \[
        \omega\in \sigma(\mc S_\infty) \iff k(\omega) \in \R.
    \]
\end{lemma}

\begin{lemma} \label{thm:bandstructure}
The spectrum $\sigma(\mc S_\infty)$ is absolutely continuous and can be decomposed into \emph{bands}
    \[
        \sigma(\mc S_\infty) = \bigcup_{j\in \Z}[a_j,b_j] \subset \R
    \]
    with band edges given by $\dots<a_j < b_j <a_{j+1}<b_{j+1}<\dots$.
\end{lemma}
A consequence of Lemma \ref{lem:inspec_iff_inband} and Lemma \ref{thm:bandstructure} is that $\omega \in \C \setminus \R$ implies $\Im k(\omega) \neq 0$, and therefore $\Tilde{P}(\omega)$ must have two eigenvalues of distinct magnitudes. 

The spatial growth or decay of resonant solutions is controlled by the imaginary part of the quasimomentum $k(\omega)$. We therefore now proceed to characterise its asymptotic behaviour as $M\to \infty$.
\begin{lemma}
    Consider a sequence of resonant frequencies $\omega_M$ of $\mc S_M$ such that 
    $\omega_M\to \omega^*\in (a_j,b_j)$ for some band $[a_j,b_j]$.
 Then, there exist constants $C_1>0$ and $r\in \Z_{\geq 0}$ such that 
    \[
        \abs{\Im \omega_M}, \abs{\Im k(\omega_M)}\leq \frac{C_1}{M}+r\frac{\ln M}{M}.
    \]
    Moreover, we have $r=0$ for almost every $\omega^*\in (a_j,b_j)$.
\end{lemma}
\begin{proof}
Let $K$ be a compact neighbourhood of $\omega^*$ that does not contain any band edges. We can choose $K$ so that $k(\omega)$ is analytic with $\frac{\dd k}{\dd \omega}(\omega)$ being nonzero on $K$. Moreover, for $\omega\in K \cap\R$, both $k(\omega)$ and its derivative are real.

Writing $\omega=x+\i y$ we may use the analytic expansion of $k(\omega)$ on $K$ to find
\[
\Im k(x+\i y) = yH(x,y)
\]
where $H(x,y)$ is a real analytic function for real $x$ and $y$. Moreover, using l'Hôpital's rule we find that $H(\omega^*, 0) = \frac{\dd k}{\dd \omega}(\omega^*) \neq 0$. After possibly shrinking $K$, we can thus bound the analytic $H(x,y)$ from below and above by $B_1, B_2>0$. Because $y = \Im \omega$ this yields the asymptotic equivalence of imaginary parts
\[
    B_1\abs{\Im \omega}\leq\abs{\Im k(\omega)} \leq B_2\abs{\Im \omega},
\]
allowing us to characterise their convergence interchangeably.
In the following, without loss of generality, we assume that $\Im \omega, \Im k(\omega)<0$.

Since $K$ is chosen to be away from the band edges and $\tp(\omega)$ has distinct eigenvalues in $\C \setminus \R$, we know that $\tp(\omega)$ is diagonalisable on $K$, and we obtain that
    \[
        \tp(\omega) = B(\omega)\begin{pmatrix}
            e^{\i k(\omega)} & 0\\
            0 & e^{-\i k(\omega)}
        \end{pmatrix} B^{-1}(\omega),
    \]
    where the matrices $B(\omega)$ and $B^{-1}(\omega)$ are analytic in $\omega \in K$.

    From \cref{lem:symm_propmat_is_propmat_of_symm}, it follows that $\omega_M$ is a resonant frequency of $\mc S_M$ if and only if 
    $\tp^M(\omega)(1,-\i\omega)^\top = \alpha(1,\i\omega)^\top$ for some $\alpha\in \C$. This is equivalent to determining the zeros of 
    \begin{align*}
        f_M(\omega) & \coloneqq \det \left( \tp^M(\omega)\begin{pmatrix}
            1\\-\i\omega
        \end{pmatrix}\middle|\begin{pmatrix}
            1\\\i\omega
        \end{pmatrix}\right)\\
       & =\det B(\omega)\det \left( \begin{pmatrix}
            e^{\i k(\omega)M} & 0\\
            0 & e^{-\i k(\omega)M}
        \end{pmatrix} B^{-1}(\omega)\begin{pmatrix}
            1\\-\i\omega
        \end{pmatrix}\middle|B^{-1}(\omega)\begin{pmatrix}
            1\\\i\omega
        \end{pmatrix}\right).
    \end{align*}
    We now define 
    \[
        \bm v_\pm(\omega) \coloneqq B^{-1}(\omega)\begin{pmatrix}
            1\\\pm\i\omega
        \end{pmatrix},
    \]
    which are also analytic in $\omega$. 
    Writing $c_1(\omega) \coloneqq \seq{\bm v_-}{1}(\omega)\seq{\bm v_+}{2}(\omega)$ and $c_2(\omega) \coloneqq \seq{\bm v_-}{2}(\omega)\seq{\bm v_+}{1}(\omega)$, we thus find that $\omega_M$ is a resonant frequency of $\mc S_M$ if and only if it is a zero of the function 
    \[
        g_M(\omega) \coloneqq c_1(\omega)e^{\i k(\omega) M}-c_2(\omega)e^{-\i k(\omega) M}.
    \]
    Here, we have used the fact that $\det B(\omega)$ is nonzero on $K$.
    
    We can now distinguish two cases: Either $c_1(\omega^*), c_2(\omega^*)\neq 0$ or either 
    $c_1(\omega^*)=0$ or $c_2(\omega^*)=0$. The case $c_1(\omega^*)=c_2(\omega^*)=0$ can be excluded after noticing that 
    \[ 
        c_1(\omega^*) - c_2(\omega^*) = \det \left(B^{-1}(\omega^*) \begin{pmatrix}
            1 & 1\\
            -\i \omega^* & \i \omega^*
        \end{pmatrix}\right) \neq 0.
    \]
    Moreover, because  $c_1(\omega)$ and $c_2(\omega)$ are analytic functions, their zeros are isolated. Without loss of generality, after assuming that $c_2(\omega^*) \neq 0$, we may thus shrink the neighbourhood $K$ so that $c_2(\omega) \neq 0$ on $K$,  $c_1(\omega) \neq 0$ on $K\setminus \{\omega^*\}$ and $\omega^*$ is potentially a zero of $c_1(\omega)$ with finite multiplicity. We now proceed by treating these cases separately.
    
    \textbf{Case} $c_1(\omega^*)\neq 0$.\\
    In this case, $c_1$ and $c_2$ are bounded from above and below on $K$. Applying the reverse triangle inequality yields
    \begin{gather*}
        \abs{g_M(\omega)} > \abs{\; \abs{c_1(\omega)e^{\i k(\omega) M}} - \abs{c_2(\omega)e^{-\i k(\omega) M}} \;} > C_1 e^{-M\Im (k(\omega))} - C_2.
    \end{gather*} 
    Because, $\Im k(\omega) < 0$, $g_M(\omega) =0$ is thus only possible as $M\to \infty$ if $M\Im (k(\omega))$ remains bounded, proving the claim with $r=0$ in this case. Because the zeros of $c_1(\omega)$ are isolated, this is the generic case, i.e. $c_1(\omega^*) \neq 0$ holds almost surely.

    \textbf{Case} $c_1(\omega^*)= 0$.\\
    In this case, we may additionally assume that there exists a subsequence $\omega_N$ of $\omega_M$ such that
    \[
        \abs{\omega_N-\omega^*} > \frac{1}{N}.
    \]
    Otherwise, the claim follows immediately. Now, because $\omega^*$ is the only zero of $c_1(\omega)$ on $K$, with some finite multiplicity $r$ we find that
    \[
        \abs{c_1(\omega_N)} > \frac{C_1}{N^r}.
    \]
    Repeating the above arguments, we thus obtain that
    \begin{gather*}
        \abs{g_N(\omega_N)} > \abs{\abs{c_1(\omega_N)e^{\i k(\omega_N) N}} - \abs{c_2(\omega_N)e^{-\i k(\omega_N) N}}} > \frac{C_1}{N^r} e^{-N\Im (k(\omega_N))} - C_2.
    \end{gather*}
    Therefore, in order for $g_N(\omega_N)=0$ to be possible, $\abs{N\Im k(\omega_N)-r\ln N}$ must remain bounded, yielding the desired claim.
\end{proof}

The following result now enables us to prove the delocalisation of the resonant modes of $\mc S_M$.
\begin{theorem}\label{thm:deloc}
    Let $\omega_M\to \omega^*$ be defined as above with resonant modes $v_M\in L^2([0,ML], \dd\mu)$ such that $\norm{(v_M(0), v'_M\rvert_{\iL}(0))^\top}_2 = 1$. 
    Then, there exist constants $C_3>0$ and $r\in \Z_{\geq 0}$ such that
    \[
        (C_3M^r)^{-1} \leq \norm{(v_M(x), v'_M(x))^\top}_2 \leq C_3M^r \quad \text{for }x\in [0,ML].
    \]
    Moreover, we have $r=0$ for almost every $\omega^*\in (a_j,b_j)$.
\end{theorem}
\begin{proof}
    Absorbing constant factors into $C_3$, we equivalently prove that $(C_2M^r)^{-1} \leq \sigma_1 \leq \sigma_2 \leq C_2M^r$ where $\sigma_1,\sigma_2$ are the two singular values of $\tp^n(\omega_M)$ for some $0\leq n \leq M$. This is due to the fact that $\tp^n(\omega_M)$ propagates and therefore controls the magnitude of the vector 
    \[
        \begin{pmatrix}
            v_M(nL)\\ v_M'(nL)
        \end{pmatrix}= \tp^n(\omega)\begin{pmatrix}
            1\\ -\i \omega
        \end{pmatrix}.
    \]
    Furthermore, because $\sigma_1\sigma_2 = \det \tp^n(\omega)=1$, it suffices to find an upper bound for $\sigma_2$.

    We may assume without loss of generality that $\Im \omega_M<0$ and $k(\omega_M)<0$. We then write $k(\omega_M) = \alpha_M -\i \beta_M$ and have $0< \beta_M \leq (C_1 + r\ln M)/M$, by the previous lemma.

    We again use the diagonalisability of $\tp^n(\omega_M)$ to obtain
    \[
        \tp^n(\omega_M) = B(\omega_M)\begin{pmatrix}
            e^{n\beta_M}e^{\i n\alpha_M} & 0\\
            0 & e^{-n\beta_M}e^{-\i n\alpha_M}
        \end{pmatrix} B^{-1}(\omega_M),
    \] and hence we also have 
    \[
    \sigma_2 = \norm{\tp^n(\omega)} \leq \abs{e^{n\beta_M}}\norm{B(\omega)}\norm{B^{-1}(\omega)} \leq e^{n(C_1+r\ln M)/M}C_dC_e \leq C_2M^r.
    \]
    Here, we again used the fact that $\tp(\omega)$ is diagonalisable on $K$, making $B(\omega)$ and $B^{-1}(\omega)$ analytic, and absorbed all constants into $C_2>0$.
\end{proof}

\begin{figure}
    \centering
    \includegraphics[width=0.6\linewidth]{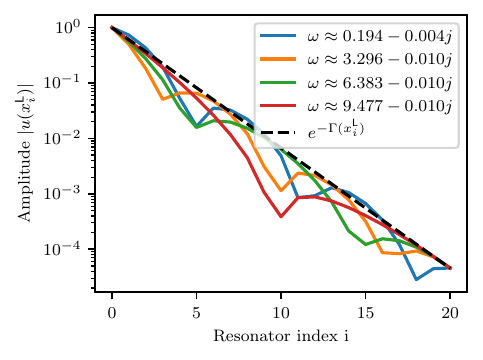}
    \caption{Resonant modes $u(x)$ corresponding to a variety of subwavelength and non-subwavelength resonant frequencies $\omega$ for a single-resonator unit cell ($N=1, \ell=s=1$) repeated periodically $M=20$ times. The dashed line denotes the characteristic decay $e^{-\Gamma(x)}$.}
    \label{fig:decay_comparison}
\end{figure}

Finally, we connect the delocalisation of the resonant modes $v_M(x)$ of $\mc S_M$ to the exponential localisation of the eigenmodes $u_M(x)$ of \cref{eq:skin_coupled} with the material parameters $\rho_M, \kappa_M,$ and $\gamma_M$. In particular, due to the fact that $v_M(x)$ can grow at most polynomially and $u_M(x) = e^{-\Gamma(x)}v_M(x)$, we have the following result. 
\begin{corollary}\label{cor:skin_decay}
    Let $\omega_M\to \omega^*$ as above and $u_M\in L^2([0,ML], \dd\mu)$ be the corresponding eigenmode of \cref{eq:skin_coupled} with $\norm{(u_M(0), u'_M\rvert_{\iL}(0))^\top}_2 = 1$. Assume further that $\gamma_i >0$ for all $i=1,\dots, N$. 
    
    In this case, $u_M$ is exponentially localised at the left edge. 
    More precisely, we have 
    \[
        \abs{u_M(x)} \leq C_3M^re^{-\Gamma(x)},
    \]
    where $C_3$ and $r$ are as in \cref{thm:deloc}.
\end{corollary}

This corollary is illustrated in \cref{fig:decay_comparison} for the resonant modes of a  one-resonator unit cell ($N=1$), repeated $M=20$ times. Indeed, both the subwavelength and non-subwavelength resonant modes $u(x)$ decay exactly proportionally to the characteristic decay $e^{-\Gamma(x)}$, as predicted.

We note that using the approach introduced in \cite{ammari2025competingedgebulklocalisation}, we can generalise our proof here for the skin effect in the non-subwavelength regime to aperiodic systems.

\subsection{Damping transition}
Finally, in this subsection, our aim is to investigate another question that arises from the form \cref{eq:skin_coupled} within the resonators
\[
    \frac{{\dd}^{2} u}{\dd x^2}(x) + 2\gamma_i \frac{\dd u}{\dd x}(x) +\omega^2u(x)=0,
\]
where we assumed $v_i=1$.
The general solution is given by 
\begin{equation}\label{eq:fundamental_soln}
    u(x) = e^{-\gamma_i x}(ae^{\nu x}+be^{-\nu x}),
\end{equation}
where $\nu = \sqrt{\gamma_i^2-\omega^2}$ as in \cref{lem:propagation}. 
Thus, $\nu$ transitions from being real to purely imaginary at $\gamma_i=\omega$ which one would expect to correspond to a transition from decaying to oscillatory behaviour within the resonators. Indeed, $\gamma_i$ corresponds exactly to a spatial damping term, and this transition is the underdamping--overdamping transition of the damped harmonic oscillator.

We demonstrate this transition empirically in \cref{fig:nonreciprocal_damping} by plotting the symmetrised resonant modes $v(x)$ given by \eqref{symm} for two dimer systems with distinct decay rates $\gamma$. We observe that the resonant modes undergo a transition from oscillatory to non-oscillatory within the resonators, as $\gamma$ is increased while $\omega$ remains roughly constant.

Nevertheless, we would like to note that, due to the $e^{-\gamma_i x}$ prefactor, the overall decay behaviour of the non-symmetrised $u(x)$ is still dominated by the characteristic decay regardless of which side of the transition $\omega$ falls on.
\begin{figure}
    \centering
    \includegraphics[width=0.9\linewidth]{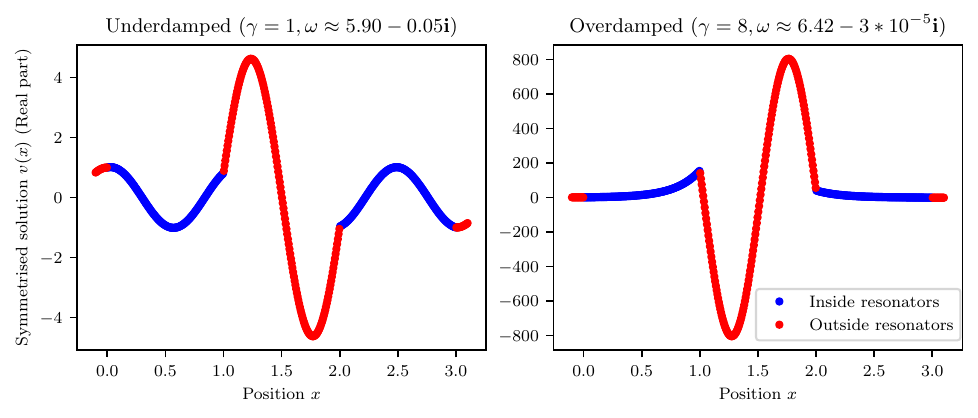}
    \caption{Symmetrised resonant modes $v(x)$ at a resonant frequency $\omega$ for two dimer systems with distinct values of $\gamma_1 =\gamma_2 = \gamma$ ($\ell_i = s_i=v_i = 1, \delta=10^{-1}$). While for low $\gamma=1 < \omega$ the mode is oscillatory within the resonators, this is no longer the case for high $\gamma=8>\omega$.} 
    \label{fig:nonreciprocal_damping}
\end{figure}

\section*{Acknowledgments}
    This work was supported by the Swiss National Science Foundation grant number 200021-236472 (AU and JQ), the Fundamental Research Funds for the Central Universities grant number 226-2025-00192 (PL and YS), and the National Key R\&D Program of China grant number 2024YFA1016000 (BL and PL). It was initiated while H.A. was visiting the Hong Kong Institute for Advanced Study as a Senior Fellow. 
\section*{Code availability}
The software used to produce the numerical results in this work is openly available at \\ \href{https://doi.org/10.5281/zenodo.18622435}{https://doi.org/10.5281/zenodo.18622435}.

\printbibliography

@article{feppon.cheng.ea2023Subwavelength,
  title        = {Subwavelength Resonances in One-Dimensional High-Contrast Acoustic Media},
  author       = {Feppon, Florian and Cheng, Zijian and Ammari, Habib},
  date         = {2023},
  journaltitle = {SIAM J. Appl. Math.},
  volume       = {83},
  number       = {2},
  pages        = {625--665},
}

@mist{pm1,
    title = {Resonance Analysis of One-Dimensional Acoustic Media: A Propagation Matrix Approach},
  shorttitle = {Resonance Analysis of One-Dimensional Acoustic Media},
  author = {Huang, Yi and Li, Bowen and Liu, Ping and Shao, Yingjie},
  year = {2025},
  number = {arXiv:2512.01734},
  publisher = {arXiv},
  doi = {10.48550/arXiv.2512.01734},
  url = {http://arxiv.org/abs/2512.01734},
}

@unpublished{pm2,
    author = {Ammari, Habib and Li, Bowen and Liu, Ping and Shao, Yingjie},
    title = {Generalized capacitance matrix approximation for Fabry-Pérot-type Resonances in one-dimensional
high-contrast media},
    note = {Manuscript in preparation},
    year = {2026}
}

@article {exceptpoint1,
    AUTHOR = {Ammari, Habib and Davies, Bryn and Hiltunen, Erik Orvehed and
              Lee, Hyundae and Yu, Sanghyeon},
     TITLE = {Exceptional points in parity-time-symmetric subwavelength
              metamaterials},
   JOURNAL = {SIAM J. Math. Anal.},
  FJOURNAL = {SIAM Journal on Mathematical Analysis},
    VOLUME = {54},
      YEAR = {2022},
    NUMBER = {6},
     PAGES = {6223--6253},
       DOI = {10.1137/22M1469821},
       URL = {https://doi.org/10.1137/22M1469821},
}

@article{SANCHEZSOTO2012191,
title = {The transfer matrix: A geometrical perspective},
journal = {Physics Reports},
volume = {513},
number = {4},
pages = {191-227},
year = {2012},
note = {The transfer matrix: A geometrical perspective},
doi = {https://doi.org/10.1016/j.physrep.2011.10.002},
url = {https://www.sciencedirect.com/science/article/pii/S0370157311002560},
author = {Luis L. Sánchez-Soto and Juan J. Monzón and Alberto G. Barriuso and José F. Cariñena}
}

@article {exceptpoint2,
    AUTHOR = {Ammari, Habib and Davies, Bryn and Hiltunen, Erik Orvehed and
              Lee, Hyundae and Yu, Sanghyeon},
     TITLE = {High-order exceptional points and enhanced sensing in
              subwavelength resonator arrays},
   JOURNAL = {Stud. Appl. Math.},
  FJOURNAL = {Studies in Applied Mathematics},
    VOLUME = {146},
      YEAR = {2021},
    NUMBER = {2},
     PAGES = {440--462},
       DOI = {10.1111/sapm.12349},
       URL = {https://doi.org/10.1111/sapm.12349},
}

@article {yokomizo,
    AUTHOR = {Yokomizo, K. and Yoda, T. and Murakami, S.},
     TITLE = {Non-hermitian waves in a continuous periodic
model and application to photonic crystals},
   JOURNAL = {Phys. Rev. Res.},
    VOLUME = {4},
      YEAR = {2022},
     PAGES = {023089},
       DOI = {10.1103/PhysRevResearch.4.023089},
}

@article{PhysRevResearch.3.013223,
  title = {Non-Hermitian Fabry-P\'erot resonances in a PT-symmetric system},
  author = {Shobe, Ken and Kuramoto, Keiichi and Imura, Ken-Ichiro and Hatano, Naomichi},
  journal = {Phys. Rev. Res.},
  volume = {3},
  issue = {1},
  pages = {013223},
  numpages = {15},
  year = {2021},
  doi = {10.1103/PhysRevResearch.3.013223},
  url = {https://link.aps.org/doi/10.1103/PhysRevResearch.3.013223}
}

@preamble{ "\providecommand{\noopsort}[1]{} " }

@article{ammari2025competingedgebulklocalisation,
      title={Competing edge and bulk localisation in non-reciprocal disordered systems}, 
      author={Habib Ammari and Silvio Barandun and Clemens Thalhammer and Alexander Uhlmann},
      year={2025},
      journal={Multiscale Model. Simul., to appear  (arXiv:2507.02096)},
}

@misc{ammari.barandun.ea2023NonHermitian,
  title = {The {{Non-Hermitian Skin Effect With Three-Dimensional Long-Range Coupling}}},
  author = {Ammari, Habib and Barandun, Silvio and Cao, Jinghao and Davies, Bryn and Hiltunen, Erik Orvehed and Liu, Ping},
  year = 2023,
  month = nov,
  number = {arXiv:2311.10521},
  eprint = {2311.10521},
  primaryclass = {cond-mat, physics:math-ph},
  publisher = {arXiv},
  doi = {10.48550/arXiv.2311.10521},
  url = {http://arxiv.org/abs/2311.10521},
  archiveprefix = {arXiv}
}

@article{ammari.barandun.ea2025Generalized,
  title = {Generalized {{Brillouin}} Zone for Non-Reciprocal Systems},
  author = {Ammari, Habib and Barandun, Silvio and Liu, Ping and Uhlmann, Alexander},
  year = 2025,
  month = mar,
  journal = {Proceedings of the Royal Society A: Mathematical, Physical and Engineering Sciences},
  volume = {481},
  number = {2310},
  pages = {20240643},
  publisher = {Royal Society},
  doi = {10.1098/rspa.2024.0643},
  url = {https://royalsocietypublishing.org/doi/10.1098/rspa.2024.0643}
}

@article{ammari.davies.ea2024Functional,
  title = {Functional Analytic Methods for Discrete Approximations of Subwavelength Resonator Systems},
  author = {Ammari, Habib and Davies, Bryn and Hiltunen, Erik Orvehed},
  year = 2024,
  month = oct,
  journal = {Pure and Applied Analysis},
  volume = {6},
  number = {3},
  pages = {873--939},
  doi = {10.2140/paa.2024.6.873},
}

@misc{ammari.thalhammer.ea2025Uniform,
  title = {Uniform {{Hyperbolicity}}, {{Bandgaps}} and {{Edge Modes}} in {{Aperiodic Systems}} of {{Subwavelength Resonators}}},
  author = {Ammari, Habib and Thalhammer, Clemens and Uhlmann, Alexander},
  year = 2025,
  month = sep,
  number = {arXiv:2509.22417},
  eprint = {2509.22417},
  primaryclass = {math-ph},
  publisher = {arXiv},
  doi = {10.48550/arXiv.2509.22417},
  url = {http://arxiv.org/abs/2509.22417},
  archiveprefix = {arXiv}
}

@article{ammariMathematicalFoundationsNonHermitian2024,
  title = {Mathematical {{Foundations}} of the {{Non-Hermitian Skin Effect}}},
  author = {Ammari, Habib and Barandun, Silvio and Cao, Jinghao and Davies, Bryn and Hiltunen, Erik Orvehed},
  year = 2024,
  month = apr,
  journal = {Archive for Rational Mechanics and Analysis},
  volume = {248},
  number = {3},
  pages = {33},
  issn = {1432-0673},
  doi = {10.1007/s00205-024-01976-y},
  url = {https://doi.org/10.1007/s00205-024-01976-y},
  langid = {english}
}

@book{eastham1973spectral,
  title = {The Spectral Theory of Periodic Differential Equations},
  author = {Eastham, M. S. P.},
  year = 1973,
  series = {Texts in Mathematics},
  publisher = {Scottish Academic Press},
  address = {Edinburgh},
  isbn = {978-0-7011-1936-2},
  langid = {english},
  lccn = {725264}
}

@article{fleury.sounas.ea2015invisible,
  title = {An Invisible Acoustic Sensor Based on Parity-Time Symmetry},
  author = {Fleury, Romain and Sounas, Dimitrios and Al{\`u}, Andrea},
  year = 2015,
  month = jan,
  journal = {Nature Communications},
  volume = {6},
  number = {1},
  pages = {5905},
  publisher = {Nature Publishing Group},
  issn = {2041-1723},
  doi = {10.1038/ncomms6905},
  url = {https://www.nature.com/articles/ncomms6905},
  copyright = {2015 Springer Nature Limited},
  langid = {english}
}

@article{hatano.nelson1996Localization,
  title = {Localization {{Transitions}} in {{Non-Hermitian Quantum Mechanics}}},
  author = {Hatano, Naomichi and Nelson, David R.},
  year = 1996,
  month = jul,
  journal = {Physical Review Letters},
  volume = {77},
  number = {3},
  pages = {570--573},
  publisher = {American Physical Society},
  doi = {10.1103/PhysRevLett.77.570},
  url = {https://link.aps.org/doi/10.1103/PhysRevLett.77.570}
}

@article{hodaei.hassan.ea2017Enhanced,
  title = {Enhanced Sensitivity at Higher-Order Exceptional Points},
  author = {Hodaei, Hossein and Hassan, Absar U. and Wittek, Steffen and {Garcia-Gracia}, Hipolito and {El-Ganainy}, Ramy and Christodoulides, Demetrios N. and Khajavikhan, Mercedeh},
  year = 2017,
  month = aug,
  journal = {Nature},
  volume = {548},
  number = {7666},
  pages = {187--191},
  publisher = {Nature Publishing Group},
  issn = {1476-4687},
  doi = {10.1038/nature23280},
  url = {https://www.nature.com/articles/nature23280},
  copyright = {2017 Macmillan Publishers Limited, part of Springer Nature. All rights reserved.},
  langid = {english}
}

@book{kato1995Perturbation,
  title = {Perturbation Theory for Linear Operators},
  author = {Kat{\=o}, Tosio},
  year = 1995,
  series = {Classics in Mathematics},
  publisher = {Springer},
  address = {Berlin},
  isbn = {978-3-540-58661-6},
  langid = {english},
  lccn = {QA329.2 .K37 1995}
}

@article{rivero.feng.ea2022Imaginary,
  title = {Imaginary {{Gauge Transformation}} in {{Momentum Space}} and {{Dirac Exceptional Point}}},
  author = {Rivero, Jose H. D. and Feng, Liang and Ge, Li},
  year = 2022,
  month = dec,
  journal = {Physical Review Letters},
  volume = {129},
  number = {24},
  pages = {243901},
  publisher = {American Physical Society},
  doi = {10.1103/PhysRevLett.129.243901},
  url = {https://link.aps.org/doi/10.1103/PhysRevLett.129.243901}
}

@article{vollmer.arnold.ea2008Single,
  title = {Single Virus Detection from the Reactive Shift of a Whispering-Gallery Mode},
  author = {Vollmer, F. and Arnold, S. and Keng, D.},
  year = 2008,
  month = dec,
  journal = {Proceedings of the National Academy of Sciences},
  volume = {105},
  number = {52},
  pages = {20701--20704},
  publisher = {Proceedings of the National Academy of Sciences},
  doi = {10.1073/pnas.0808988106},
  url = {https://www.pnas.org/doi/10.1073/pnas.0808988106}
}

@book{zettl2005SturmLiouville,
  title = {Sturm-{{Liouville}} Theory},
  author = {Zettl, Anton},
  year = 2005,
  series = {Mathematical Surveys and Monographs ; Volume 121},
  publisher = {American Mathematical Society},
  address = {Providence, Rhode Island},
  isbn = {978-1-4704-1348-4},
  langid = {english}
}

@article {junshan1,
    AUTHOR = {Lin, Junshan and Zhang, Hai},
     TITLE = {Mathematical theory for topological photonic materials in one
              dimension},
   JOURNAL = {J. Phys. A},
  FJOURNAL = {Journal of Physics. A. Mathematical and Theoretical},
    VOLUME = {55},
      YEAR = {2022},
    NUMBER = {49},
     PAGES = {Paper No. 495203, 45},
}

@article {barandun2023,
    AUTHOR = {Ammari, Habib and Barandun, Silvio and Cao, Jinghao and
              Feppon, Florian},
     TITLE = {Edge modes in subwavelength resonators in one dimension},
   JOURNAL = {Multiscale Model. Simul.},
  FJOURNAL = {Multiscale Modeling \& Simulation. A SIAM Interdisciplinary
              Journal},
    VOLUME = {21},
      YEAR = {2023},
    NUMBER = {3},
     PAGES = {964--992},
       DOI = {10.1137/23M1549419},
       URL = {https://doi.org/10.1137/23M1549419},
}
\end{document}